%% file: mainlong.tex
\documentclass[a4paper,12pt]{article}

\usepackage[margin=1in,footskip=0.25in]{geometry}
\usepackage{amsthm}
\usepackage{mathtools}
\usepackage{ebproof}

\usepackage[T1]{fontenc} 
\usepackage{lmodern} 
\usepackage{amssymb} 
\usepackage{stmaryrd} 
\usepackage{cmll}
\usepackage{xcolor}
\usepackage{biblatex}

\usepackage{tikz-cd}

    \newtheorem{fact}{Fact}[section]
\newtheorem{definition}[fact]{Definition}
\newtheorem{proposition}[fact]{Proposition}
\newtheorem{theorem}[fact]{Theorem}
\newtheorem{lemma}[fact]{Lemma}

\newtheorem{example}[fact]{Example}
\newtheorem{remark}[fact]{Remark}

\title{Coherence by Normalization for Linear Multicategorical Structures} 

\author{Federico Olimpieri}

\input{macro.tex}

\begin{document}

\maketitle

\begin{abstract} 
We establish a formal correspondence between resource calculi an appropriate linear multicategories. We consider the cases of (symmetric) representable, symmetric closed and autonomous multicategories. For all these structures, we prove that morphisms of the corresponding free constructions can be presented by means of typed resource terms, up to a reduction relation and a structural equivalence. Thanks to the linearity of the calculi, we can prove strong normalization of the reduction by combinatorial methods, defining appropriate decreasing measures. From this, we achieve a general coherence result: morphisms that live in the free multicategorical structures are the same whenever the normal forms of the associated terms are equal. As further application, we obtain syntactic proofs of Mac Lane's coherence theorems for (symmetric) monoidal categories.
\end{abstract}

\section{Introduction} 

The basis of the celebrated \emph{Curry-Howard-Lambek correspondence} is that  logical systems, typed $\lambda$-calculi and appropriate categorical constructions are different presentations of the same mathematical structure. An important consequence of the correspondence is that we can give \emph{syntactical presentations} of categories, that can be exploited to prove general results by means of elementary methods, such as induction. At the same time, we can use categorical methods to obtain a more modular and clean design of programming languages. The classic example is given by simply typed $ \lambda$-calculi and cartesian closed categories \cite{lam:hocl}. The idea is well-known: morphisms in free cartesian closed categories over sets are identified with equivalence classes of $\lambda $-terms up to $\beta \eta  $-equality. Another important setting is the \emph{linear} one, where we consider \emph{monoidal} categories instead of cartesian ones. In this case, \emph{linear logic} \cite{gir:ll} enters the scene: symmetric monoidal closed categories correspond to \emph{linear} $\lambda $-calculi. Computationally, this is a huge restriction, since linear terms can neither copy nor delete their inputs during computation.  A refinement of this picture can be obtained by switching from categories to \emph{multicategories} \cite{lam:multi}. These structures were indeed first introduced by Lambek to achieve a categorical framework formally closer to typed calculi/proof systems. Morphisms of multicategories can have multiple sources $ f : \ty_1, \dots, \ty_n \to \ty $, recalling the structure of a \emph{type judgment} $ x_1 : \ty_1, \dots, x_n : \ty_n \vdash f : \ty.  $ 

We are interested in establishing a Curry-Howard-Lambek style correspondence for appropriate linear multicategories and then employ it to obtain \emph{coherence results}. When we deal with complex structures such as tensor products, it becomes crucial to have a \emph{decision process} to establish whether two arrows are equal. This is called a \emph{coherence} problem. The main example is Mac Lane's original result \cite{mac:coh}, which states that \emph{all structural diagrams} in monoidal categories commute. If one considers more complex structures, the class of commutative diagrams is normally more restrictive. In the case of closed monoidal categories, Kelly and Mac Lane \cite{kellymac:graph} associated  \emph{graphs} to structural morphisms, obtaining the following coherence result: two structural arrows between appropriate objects\footnote{A restriction on the type of morphisms is needed due to the presence of the monoidal unit.} are equal whenever their graph is the same. We aim to achieve coherence results for linear multicategories, building on Lambek's and Mints \cite{mints:closed} intuition that coherence problems can be rephrased in the language of proof theory and obtained by exploiting appropriate notions of \emph{normalization} for proofs/terms \cite{lam:multi}. We do so by establishing  a formal connection between \emph{resource} calculi and \emph{linear} multicategorical structures. 

\subparagraph{Main Results} We study free multicategorical constructions for (symmetric) representable and closed structures. Representability consists of the multicategorical monoidal structure \cite{her:rep}. We prove that free linear multicategories built on appropriate signatures can be presented by means of typed resource calculi, where morphisms correspond to equivalence classes of terms up to a certain equivalence. We handle the tensor product \textit{via} \emph{pattern-matching}, presented as a syntactic \emph{explicit substitution}. The definition of our type systems is given in \emph{natural deduction} style: we have introduction and elimination rules for each type constructor. Our work is conceptually inspired by an `\emph{adjoint functors} point-of-view'. A basic fact of the classic Curry-Howard-Lambek correspondence is that $\beta \eta $-equality can be expressed by means of the \emph{unit} ($ \eta$) and the \emph{counit} ($ \beta $) of the adjunction between products and arrow types. We generalize this observation to the multicategorical setting, thus introducing an appropriate reduction relation that corresponds to the representable structure.  Indeed, a fundamental aspect of our work consists of the in depth study of resource terms rewriting. We introduce confluent and strongly normalizing reductions, that express the appropriate equalities. In order to do so, we exploit \emph{action-at-distance} to define our operational semantics, that has proven to be a successful approach to calculi with explicit substitution \cite{kes:expl, ben:sub, ben:stand}. An important feature of this approach is to distinguish the operational semantics, defined by action-at-distance, from a notion of \emph{structural equivalence}, that deals with commutations of explicit substitution with the other syntactic constructors. This approach overcomes the classic difficulties of rewriting systems with explicit substitution, allowing us to obtain confluence and strong normalization in an elegant way. In this way, we get a general \emph{coherence result}: two structural morphisms of linear multicategories are equal whenever the \emph{normal forms} of their associated terms are equal. In the context of (symmetric) representable multicategories, we apply this result to obtain a \emph{syntactic proof} of stronger coherence theorems, that can be seen as multicategorical versions of the classic MacLane coherence theorems for (symmetric) monoidal categories \cite{mac:coh}. The coherence theorem for representable multicategories was already proved in \cite{her:rep}. We give an alternative type-theoretic proof for it. To our knowledge, the other coherence results that we present are new. Moreover, exploiting the equivalence between monoidal categories and representable multicategories established by Hermida \cite{her:rep}, we are able to obtain the original Mac Lane's results as corollaries of our coherence theorems.


\subparagraph{Related Work} Building on Lambek's original ideas, several researchers have advocated the use of multicategories to model computational structures. Hyland \cite{hyland:classical} proposed to rebuild the theory of pure $\lambda $-calculus by means of \emph{cartesian operads}, that is one-object cartesian multicategories. The idea of seeing resource calculi as multicategories was first employed by Mazza \emph{et al.} \cite{mazza:pol, mazza:hdr}. We build on their approach, showing that these calculi correspond to appropriate \emph{universal constructions}, namely free linear multicategories. The first resource calculus has been introduced by Boudol \cite{boud:res}. A similar construction was also independently considered by Kfoury \cite{kfo:res}. Resource terms have gained special interest thanks to the definition by Ehrhard and Regnier of the \emph{Taylor expansion} for $\lambda $-terms \cite{er:tay}. From this perspective, the resource calculus is a \emph{theory of approximation of programs} and has been successfully exploited to study the computational properties of $\lambda $-terms \cite{barbaro:tay, vaux:tay, mazza:pol, ol:thesis}. Our syntax is very close to the one of \emph{polyadic calculi} or \emph{rigid resource calculi} \cite{mazza:pol, tao:gen}. We need to extend the standard operational semantics, adding an $\eta $-reduction and a reduction for explicit substitution. Our $\eta $-reduction is built from an expansion rule instead of a contraction, since $ \eta$-expansion naturally fits the adjoint point-of-view, corresponding the the unit of the considered adjunction. In dealing with the technical rewriting issues, we follow \cite{mints:closed, jay:eta, kes:fix}, obtaining a terminating $\eta $-reduction. As already discussed, we handle the explicit substitution following Accattoli and Kesner methodology \cite{kes:expl, ben:stand, ben:sub}. 

 The calculi we present are also strongly related to \emph{intuitionistic linear logic} \cite{depaiva:ill}. It is well-known that resource calculi can be seen as fragments of ILL \cite{mazza:pol, mazza:hdr}. While ILL is presented \textit{via} sequent calculus, we chose a natural deduction setting, this latter being directly connected to the `adjoint functors' point-of-view. Accattoli and Kesner approach to explicit substitution allows us to bypass the cumbersome commutation rules needed for ILL
 rewriting. Moreover, resource calculi are closer to the multicategorical definitions (their constructors being \emph{unbiased} \cite{lein:high}, \textit{i.e.}, $k $-ary). Our handling of symmetries is also more canonical and explicit. We use the properties of \emph{shuffle permutations}, in a way similar to Hasegawa \cite{has:mell} and Shulman \cite{shu:tt}, also inspired by our ongoing work on \emph{bicategorical semantics} \cite{ol:hom}. In this way, the type system is \emph{syntax directed} and we are able to prove that, given a term, there exists at most one type derivation for it. The pioneering work of Mints \cite{mints:closed} is very close to our perspective. Mints introduced a linear $ \lambda$-calculus to study the coherence problem of closed category by the means of normalization. We build on that approach, extending it to several different structures and to the multicategorical setting.
 
Shulman's type theory for (symmetric) monoidal categories \cite{shu:tt} does not employ explicit substitutions, being able to handle tensors in way similar to what happens with standard product types. Our proposal differs considerably from Shulman's, both in purpose and in implementation. While Shulman's goal is to start from the categorical structure and define a `practical' type theory to make computations, ours consists of establishing a formal correspondence between two \emph{independent} worlds: resource calculi and linear multicategories and then employ it to prove results about the categorical structure. 

 
\emph{Graphical approaches} to monoidal structures \cite{sel:graph} have been widely developed. Particularly interesting for our work are the Kelly-Mac Lane graphs \cite{kellymac:graph}, This approach has been extended \textit{via} linear logic, thanks to the notion of \emph{proof-net} \cite{blute:coh, hughes:star}. However, the handling of monoidal units needs extra care from this perspective, while the terms calculi approach can account for them without any particular complication.  


\section{Preliminaries}\label{sec:prel} We introduce some concepts, notations and conventions that we will use in the rest of the paper.
\subparagraph{Integers, Permutations and Lists} For $  n \in \mathbb{N} ,$ we set $ [n] = \{1, \dots, n \}$ and we denote by $ S_n $ the symmetric group of order $ n . $ The elements of $ S_n $ are permutations, that we identify with bijections $ [n] \cong [n] . $ Given $ \sigma, \tau \in S_n ,$ we denote by $ \sigma \circ \tau $ their composition. Given $ \sigma \in S_n , \tau \in S_m$ we denote by $ \sigma \oplus \tau : [n + m] \cong [n + m]  $ the evident induced permutation. We now introduce the notion of shuffle permutation, that is crucial to obtain canonical type derivations for resource terms with permutations (Proposition \ref{prop:symcano}).
\begin{definition}[Shuffles]
Let $ n_1, \dots, n_k \in \mathbb{N}  $ with $ n = \sum_{i = 1}^{k} n_i .$ A $ (n_1, \dots, n_k) $-\emph{shuffle} is a bijection $ \sigma : \sum_{i = 1}^{k} [n_i] \cong [n] $ such that the composite $   [n_i] \hookrightarrow   \sum_{i = 1}^{k} [n_i] \cong [n] $ is monotone for all $ i \in [k] . $ We denote the set of all $ (n_1, \dots, n_k) $-shuffles as $ \mathsf{shu}(n_1, \dots, n_k). $
\end{definition}


The relevant result on shuffles is the following, that induces canonical decomposition of arbitrary permutations over sums of integers.
\begin{lemma}    \label{shucan}
Every permutation $ \sigma \in S_{\sum_{i = 1}^{k} n_i} $ can be  canonically decomposed as $ \tau_0 \circ ( \bigoplus_{i = 1}^{k} \tau_i ) $ with $ \tau_0 \in  \mathsf{shu}(n_1, \dots, n_k)$ and $ \tau_i \in S_{n_i} $ for $ i \in [k] . $
\end{lemma}
Given a set $A $ and a list of its elements $ \gamma = \ty_1, \dots, \ty_k $ and $ \sigma \in S_k $ we set $ \act{\gamma}{\sigma} = a_{\sigma(1)}, \dots, \ty_{\sigma(k)} $ for the symmetric group right action. We write $\length{\gamma} $ for its length. We denote the \emph{stabilisers} for this action as $ \stabi{\gamma} = \{  \sigma \in S_k \mid   \act{\gamma}{\sigma}= \gamma  \} . $ Given lists $\gamma_1, \dots, \gamma_k ,$ we set $ \mathsf{shu}(\gamma_1, \dots, \gamma_k) = \mathsf{shu}(\length{\gamma_1}, \dots, \length{\gamma_k}).$




\subparagraph{Multicategories} Multicategories constitute the main object of our work. A multicategory is a multigraph that comes equipped with an appropriate composition operation.

\begin{definition}A \emph{multigraph} $\mathcal{G}  $ is given by the following data:\begin{itemize}\item A collection of \emph{nodes} $ \mathcal{G}_0  \ni a, b, c \dots  $\item For every $ a_1, \dots, a_n, b \in \mathcal{G}_0 ,$ a collection of \emph{multiarrows} $ \mathcal{G}(\ty_1, \dots, \ty_n; b) \ni s, t, u \dots $\end{itemize} We denote by $\morph{\mathcal{G}} $ the set of all multiarrows of $\mathcal{G}. $
 \end{definition}
 
\begin{definition}
A \emph{multicategory} is a multigraph $ \mathcal{G} $ equipped with the following additional structure: 
\begin{itemize}
\item  A \emph{composition operation} $- \circ \seq{-, \dots, -} :\mathcal{G}(\ty_1, \dots, \ty_n; b) \times \prod_{i = 1}^{n} \mathcal{G}(\gamma_i, \ty_i) \to \mathcal{G}(\gamma_1, \dots, \gamma_n; \ty) .  $

\item \emph{identities}  $ id_a \in \mathcal{G}(\ty, \ty) . $

\end{itemize}

The former data is subjected to evident associativity and identity axioms. We call \emph{objects} the nodes of $ \mathcal{G} $ and \emph{morphisms} its multiarrows. 

\end{definition}

A multicategory can be equipped with structure. We now introduce the notions of \emph{symmetric}, \emph{closed} and \emph{representable} multicategories.

\begin{definition} A multicategory $ \mathcal{M} $ is \emph{symmetric} if, for $ \sigma \in S_k $ we have a family of bijections $   \act{-}{\sigma}  : \mathcal{M}(\gamma, \ty_1, \dots, \ty_k; \ty) \cong \mathcal{M}(\gamma, \ty_{\sigma(1)}, \dots, \ty_{\sigma (k)}; \ty)    $ that satisfies additional axioms \cite{lein:high}.\end{definition}

\begin{definition}A (right) \emph{closed} structure for a multicategory $ \mathcal{M} $ is given by a family of objects $ \tens{\ty}{1}{k} \multimap \ty \in \mathcal{M}$ and arrows $ ev_{\ty_1, \dots, \ty_k, \ty} : \ty_1,\dots, \ty_k, \tens{\ty}{1}{k}\multimap \ty \to \ty $ , for $ \ty_1, \dots, \ty_k, \ty \in \mathcal{M}, $ such that the maps
\[  ev \circ \seq{-, id_{a_1}, \dots, id_{a_k}} :  \mathcal{M}(\gamma ; \tens{\ty}{1}{k} \multimap \ty) \to \mathcal{M}(\gamma, \ty_1, \dots, \ty_k ; \ty)              \]induce a bijection, multinatural in $\gamma $ and natural in $\ty . $ We write $\lambda (-) $ to denote the inverses to these maps.  \end{definition}

\begin{definition}
A \emph{representable} structure for a multicategory $ \mathcal{M} $ is given by a family of objects $ \tens{\ty}{1}{k} \in \mathcal{M} $ and arrows $ \mathsf{re}_{\ty_1, \dots, \ty_k} : \ty_1,\dots, \ty_k \to \tens{\ty}{1}{k} $, for $ \ty_1, \dots, \ty_k\in \mathcal{M}, $ such that he maps\[ - \circ \seq{id_\gamma, \mathsf{re}, id_\delta
} :  \mathcal{M}(\gamma, \tens{\ty}{1}{k} , \delta ; \ty) \to \mathcal{M}(\gamma, \ty_1, \dots, \ty_k, \delta ; \ty)              \]induce a bijection, multinatural in $\gamma, \delta $ and natural in $\ty .$ We write $\mathsf{let}(-) $ to denote the inverses to these maps.\end{definition}

We use the name \emph{autonomous multicategories} to denote symmetric representable closed multicategories. We have categories of representable multicategories ($ \mcatir $), symmetric representable multicategories ($ \mcatirs $), closed multicategories ($ \mcatic $) and autonomous multicategories ($\mauto $), whose morphisms are functors that preserve the structure on the nose.

\subparagraph{Signatures} We introduce signatures for the structures we consider.

\begin{definition}
A \emph{representable signature} is a pair $\seq{\atm, \mathcal{R}} $ where  $\atm $ is a set of atoms $\atm $ and $\mathcal{R} $ is a multigraph with nodes generated by the following inductive grammar: 
\[  \mathcal{R}_0 \ni \ty ::= o \in \atm \mid \tens{\ty}{1}{k}  \qquad (k \in \mathbb{N}) .    \]
\end{definition}

\begin{definition}
A \emph{closed  signature} $ \mathcal{L} $ is a pair $\seq{\atm, \mathcal{L}} $ where  $\atm $ is a set of atoms $\atm $ and $\mathcal{L} $ is a multigraph with with nodes generated by the following inductive grammar: 
\[  \mathcal{L}_0 \ni \ty ::= o \in \atm \mid \tens{\ty}{1}{k} \multimap \ty \qquad (k \in \mathbb{N}).      \]
\end{definition}

\begin{definition}
An \emph{autonomous signature} is a pair $\seq{\atm, \mathcal{H}} $ where  $\atm $ is a set of atoms $\atm $ and $\mathcal{H} $ is a multigraph with nodes generated by the following inductive grammar: 
\[  \mathcal{H}_0 \ni \ty ::= o \in \atm \mid \tens{\ty}{1}{k} \mid \tens{\ty}{1}{k} \multimap \ty  \qquad (k \in \mathbb{N}) .    \]
\end{definition}

We shall often identify a signature with its graph. There are categories $ \mathsf{ClosedSig} , \mathsf{RepSig}$ and $ \mathsf{AutoSig} $ for, respectively, closed, representable and autonomous signatures. We have forgetful functors from the categories $ \mathsf{ClosedM}, \mathsf{RepM}  $ and $ \mauto,$ which we denote by $ (\overline{-}) . $ One of the main goals of this paper is to build the left adjoints to those functors \textit{via} appropriate resource calculi.

\subparagraph{Monoidal Categories vs Representable Multicategories}  In order to transport coherence results from (symmetric) representable multicategories to ordinary (symmetric) monoidal categories, we shall employ an equivalence result due to Hermida \cite[Theorem 9.8]{her:rep}. Let $ \Mon $ be the category of monoidal categories and lax monoidal functors.
\begin{theorem}[\cite{her:rep}]\label{eqmonrep}There is an equivalence of categories 
    $
\begin{tikzcd}[row sep=0pt, column sep=10pt]
\mathsf{RepM}  \ar[rr,bend left=20,"{\mathsf{rep} (-)}"]
  &
  \simeq
  &
  \Mon
  \ar[ll,bend left=20,"\mathsf{mon}(-)"].
\end{tikzcd}  $
\end{theorem}

The representable structure of a monoidal category $(\mathsf{\mathbb{M}}, \otimes_{\mathbb{M}}, 1) $ is given by $ (\ty_1 \otimes_{\mathsf{\mathbb{M}}} \dots  \otimes_\mathsf{\mathbb{M}} \ty_k) = (\ty_1) \otimes_{\mathbb{M}} ( \ty_2 \otimes_{\mathbb{M}} ( \dots \otimes_{\mathbb{M}} \ty_k) \dots ). $ Then composition needs a choice of structural isomorphisms of $ \mathbb{M}  $ to be properly defined \cite[Definition 9.2]{her:rep}\footnote{If we assume Mac Lane's Coherence Theorem, the choice is unique. However, we shall not do so, since we are going to exploit Theorem \ref{eqmonrep} to \emph{transport} an appropriate coherence theorem on representable multicategories to ordinary monoidal categories, thus obtaining the Mac Lane's result as corollary.}. The former equivalence can be extended to the symmetric case in the natural way.






\subparagraph{Notations and Conventions} Given a set of terms $A $ and a reduction relation $ \to_{\epsilon} \subseteq A \times A , $ we denote respectively as $\twoheadrightarrow_\epsilon $ and $  \to^{\ast}_{\epsilon} $ its transitive closure and its transitive and reflexive closure. We denote by $ =_{\epsilon} \subseteq A \times A$  the smallest equivalence relation generated by $  \to_{\epsilon} .$ For a confluent reduction, we denote by $ \NF{s}_\epsilon $ the normal form of $ s ,$ if it exists. Given an equivalence relation $ {\mathsf{e}} \subseteq {A \times A}, $ and $ s\in A, $ we denote by $ [s]_{\mathsf{e}} $ the corresponding equivalence class. We will often drop the annotation and just write $ [s].$ We fix a countable set of variables $\mathcal{V} ,$ that we will use to define each calculi. Terms are always considered up to renaming of bound variables. Given terms $ s, t_1, \dots, t_k $ and variables $ x_1, \dots, x_k $ we write $ \subst{s}{x_1, \dots, x_k}{t_1, \dots, t_k} $ to denote capture-avoiding substitutions. We often use the abbreviation $\subst{s}{\vec{x}}{\vec{t}} .$ To define reduction relations, we rely on appropriate notions of \emph{contexts with one hole}. Given a context with hole $\ctx $ and a term $s $ we write $ \ctx [s] $ for the capture-allowing substitution of the holes of $ \ctx$ by $ s$. The \emph{size} of a term $ \size{s}$ is the number of syntactic constructors appearing in its body.  The calculi we shall introduce are typed \emph{à la Church}, but we will constantly keep the typing implicit, to improve readability. Given $ \gamma \vdash s : \ty$ we write $ \ctx[\delta \vdash p : b] = s $ meaning that $\ctx[p] = s $ and the type derivation of $ \gamma \vdash p : b $ contains a subderivation with conclusion $ \delta \vdash p : b .$ Given a typing judgment $ x_1 : \ty_1, \dots, x_n : \ty_n \vdash s : \ty  $ we shall consider variables appearing in the typing context as bound and we will work up to renaming of those variables. We write $ {\pi} \triangleright {\gamma \vdash s : \ty} $ meaning that $\pi $ is a type derivation of conclusion $\gamma \vdash s : \ty .$   For any typing rule with multiple typing contexts, we assume those contexts to be disjoint. 

\section{A Resource Calculus for Representable Multicategories}\label{sec:rep}

We present our calculus for representable multicategories. We begin by introducing its syntax and typing, then we discuss its operational semantics. We prove confluence and strong normalization for its reduction. We show that equivalence classes of terms modulo reduction and a notion of \emph{structural equivalence} define the morphisms of free representable multicategories over a signature. As an application of this result, we give a proof of the coherence theorem for representable multicategories.

\begin{figure}[!t]		 \begin{gather*}  
			\begin{prooftree}  \hypo{f \in \mathcal{R}(\ty_1, \dots, \ty_n; b)  } \hypo{ \gamma_{1} \vdash s_{1} : \ty_{1} \dots \gamma_{n} \vdash s_{n} : \ty_{n}  }
\infer2{  { ( \gamma_{1}, \dots, \gamma_{n})} \vdash    f(s_1, \dots, s_n) : b }  \end{prooftree}  \qquad 
\begin{prooftree}  \hypo{ \gamma_1 \vdash s_1 : \ty_1 \dots \gamma_k \vdash s_k : \ty_k  }  \infer1{\gamma_1, \dots, \gamma_k \vdash {\seqdots{s}{1}{k}}: \tens{\ty}{1}{k} }  \end{prooftree} \\[1em] \begin{prooftree} \hypo{\ty \in \mathcal{R}_0 }\infer1{ x : \ty \vdash x : \ty }  \end{prooftree} \qquad  \begin{prooftree} \hypo{\gamma \vdash s : \tens{\ty}{1}{k}  } \hypo{  \delta, x_1 : \ty_1, \dots, x_k : \ty_k , \delta' \vdash t : b  }\infer2{  \delta, \gamma, \delta' \vdash {t [x_1^{\ty_1}, \dots, x_k^{\ty_k} := s] }: b  }  \end{prooftree}  \end{gather*} \hrulefill
\begin{align*}
&	\ctx ::= \hole{\cdot} \mid  \seq{s_1, \dots, \ctx, \dots, s_k}   \mid   \ctx [\vec{x} := t] \mid s [\vec{x} := \ctx]   \mid f (s_1, \dots, \ctx, \dots, s_k).
			\\
&	\ctxe ::= \hole{\cdot} \mid   \seq{s_1, \dots, \ctxe, \dots, s_k} \mid \ctxe [\vec{x} := s] \mid s [\vec{x} := \ctxe] \quad (\ctxe \neq \hole{\cdot} )  \mid f (s_1, \dots, \ctxe, \dots, s_k).
\\
&			 \ctxl ::= \hole{\cdot} \mid \ctxl [\vec{x} := t].
			\end{align*}
\hrulefill
	\caption{\small Representable Type System on a signature $\mathcal{R}$ and contexts with one hole. Types are the elements of $\mathcal{R}_0 . $ }
	\label{fig:rep-calc}
\end{figure}

\subparagraph{Representable Terms} Let $ \mathcal{R} $ be a representable signature. The \emph{representable resource terms} over $ \mathcal{R} $ are defined by the following inductive grammar: \[ \repTerms (\mathcal{R}) \ni s,t ::= x \in \mathcal{V}  \mid \seqdots{s}{1}{k}    \mid s[x^{\ty_1}_1, \dots, x^{\ty_k}_k := t] \mid f (s_1, \dots, s_k) \]  for $  k \in \mathbb{N}  $ and $ f \in \morph{\mathcal{R}} , \ty_i \in \mathcal{R}. $ A term of the shape $ \seqdots{s}{1}{k} $ is called a \emph{list}. A term of the shape $s[x_1, \dots, x_k := t]   $ is called an \emph{(explicit) substitution}.  Variables under the scope of an explicit substitution are bound. Given a term $ s,$ we denote by $\mathsf{ST}(s)$ the set of its \emph{subterms} defined in the natural way. 
\begin{remark}
Our calculus follows the linear logic tradition of modelling the tensor product structure by means of a $ \mathsf{let} $ constructor \cite{depaiva:ill}. We opted for the syntactic choice of an explicit substitution $s[x_1, \dots, x_k := t] , $ which stands for the more verbose $ \mathsf{let} $ expression, $ \mathsf{ let } \ \seq{x,_1, \dots, x_k} := t \ \mathsf{ in } \ s . $ Terms of the shape $ f (s_1, \dots, s_k)$ are needed to capture the multiarrows induced by the signature $ \mathcal{R}.  $
\end{remark}

Typing and contexts with hole for representable terms is defined in Figure \ref{fig:rep-calc}. A context is \emph{atomic} when it contains just atomic types. We define the following subset of terms $\mathsf{LT} = \{   \ctxl[\seqdots{s}{1}{k}] \mid \text{ for some context } \ctxl \text{ and terms } s_i        \}  .$

\begin{remark} The condition about disjoint contexts grants \emph{linearity}.  A term is \emph{linear} when each variable appears at most once in its body. It is easy to check that, by construction, all typed terms are linear. Moreover, given $ \gamma \vdash s : \ty , $ the context $\gamma $ is \emph{relevant}, meaning that it contains just the free variables of $ s.$ 
\end{remark}

 A type of the shape $ \tens{\ty}{1}{k}   $ is called a $ k$-ary \emph{tensor product}. We use a vector notation to refer to arbitrary tensors, \textit{eg.}, $ \tyl, \vec{b} \dots $ If $ k = 0 ,$ the type $ () $ is also called the \emph{unit}. We set $ \repTerms(\mathcal{R})(\ty_1, \dots, \ty_n; \ty) = \{ s \mid x_1 : \ty_1 , \dots, x_n : \ty_n \vdash_{\mathsf{rep}} s : \ty \text{ for some  } x_i \in \fv{s} .      \} . $ We observe that, given a representable term $ \gamma \vdash s : \ty  ,$ there exists a unique type derivation for it.

 \subparagraph{Terms Under Reduction} We now introduce the reduction relation for representable terms. This relation consists of the union of two different subreductions: $\beta $ and $ \eta$ reductions, defined in Figure \ref{fig:torp}. The \emph{structural equivalence} on terms is defined as the smallest congruence on terms generated by the rule of Figure \ref{fig:torp}. We assume that the context $ \ctx  $ does not bind any variable of $ t .$ 
 
 \begin{remark}
 Our $ \eta$-reduction consists of a restricted version of the standard notion of $\eta $-expansion, the restriction is needed to achieve strong normalization. We build on a well-established tradition in term rewriting \cite{mints:closed, jay:eta, kes:fix}. Unrestricted $\eta  $ is trivially non-terminating. Indeed, for $ x : (\ty \otimes b) \vdash x : (\ty \otimes b) $ we have the non-terminating chain $  x \to_\eta \seq{x,y}[x,y := z] \to_{\eta} \seq{x,y}[x,y := \seq{v,w}[v,w := z]] \to_\eta \dots   $ Hence, we need to forbid $ \eta$-reduction on the right side of a substitution term, that is exactly what the restricted $\eta $-contexts do. Moreover, there is also a problem of interaction between $\eta $ and $\beta . $  Consider $ s = \seq{x,y} $ well-typed, then we can produce the non-terminating chain $ s \to_\eta \seq{v, w} [v, w := \seq{x,y} ] \to_\beta s \to_\eta \dots   $ For this reason, the root-step of $\eta $ has to be restricted too. The presence of a substitution context in the $ \beta$-rule is an action-at-distance \cite{ben:stand}, that allows to `free' possible blocked redexes, as the following one $   x [  x := ( \seq{y} [ y := z ] )    ]     .   $ In this way, we can bypass traditional commutation rules and retrieve good rewriting properties. Structural equivalence intuitively says that explicit substitutions can `freely travel' in the body of a term.
 \end{remark}
 
 We prove that typings are preserved under reduction and structural equivalence.
 
 \begin{proposition}[Subject Reduction and Equivalence]\label{subredrep}
Let $ s \torp s' $ or $s \streq s' $ with $ \gamma \vdash s : \ty . $ then $ \gamma \vdash s' : \ty . $  \end{proposition}
\begin{proof}
The proof is by induction on $ s \torp s' $ and $s \streq s' $ and exploits an appropriate substitution lemma.
\end{proof}
 
 We now prove that the structural equivalence is a strong bisimulation for the reduction $ \torp . $ Intuitively, this means that the equivalence does not affect terms rewriting.

 \begin{proposition}
 If $ s' \streq s  $ and $ s \toaut t $ there exists a term $ t'$ s.t.  $ t' \streq t $ and  $ s' \torp t' . $
 \end{proposition}

 We show that we can associate appropriate measures to terms that decrease under reduction. For $\beta $, we just consider the size of terms. For $ \eta$, we build on Mints approach \cite{mints:closed}. We define the \emph{size} of a type by induction: $\size{o} = 0, \size{\seqdots{\ty}{1}{k}} = 1 + \sum \size{\ty_i}. $  Given $\gamma \vdash s  : \ty $ we define a set of typed subterms of $ s$: $\mathsf{EST} (s) = \{\delta \vdash p : 	\ty \mid p \in  \mathsf{ST}(s) \setminus \mathsf{LT}  \text{ s.t. } \ctxe[\delta \vdash p : \ty] = s  \text{ for some context } \ctxe  \} .$ We set $ \eta (s) =  \sum_{\delta \vdash p : \ty \in \mathsf{EST}(s)} \size{\ty}  .  $
 

 \begin{remark} The size of terms decreases under $\beta$-reduction as a consequence of \emph{linearity}. Redexes cannot be copied nor deleted under reduction, since the substitution is linear. This fact is trivially false for standard $\lambda $-calculi, where the size of terms can possibly grow during computation.  The intuition behind the $ \eta$ measure is that we are counting all subterms of $s $ on which we could perform the $\eta $-reduction. The restrictions on the shape of $ p \in \mathsf{EST}(s)$ is indeed directly derived from the ones on $\eta $-reduction. 
 \end{remark}

 \begin{proposition} \label{repdec} The following statements hold.
If $ s \to_{\beta} s' $ then $ \size {s'} < \size s $;  if $  s \to_{\eta} s'$ then $ \eta (s') < \eta (s) .$
 \end{proposition}

\begin{proposition}\label{sepsn}
The reductions $ \to_{\beta} $ and $\to_{\eta} $  are separately strongly normalizing and confluent.
\end{proposition}
\begin{proof}
Strong normalization is a corollary of the former proposition. For confluence, first one proves local confluence by induction and then apply Newman's Lemma.
\end{proof}

We want to extend the result of separate strong normalization and confluence to the whole $\torp  $-reduction. To do so, we prove that $ \beta$ and $\eta $ suitably \emph{commutes}.

\begin{proposition}\label{rep:commuta}
If $ s \to_{\beta}^\ast t \to_{\eta}^{\ast} t' $ there exists $ s'  $ s.t. $ s \to_{\eta}^\ast s' $ and $ s' \to_{\beta}^\ast t' . $ 
\end{proposition}

\begin{theorem}\label{aut:snc}
The reduction $\torp $ is confluent and strongly normalizing. 
\end{theorem}
\begin{proof}
Strong normalization is achieved by observing that any infinite reduction chain of $ \torp$ would trigger, by Proposition \ref{rep:commuta}, an infinite reduction chain for $\eta $, that is strongly normalizing. Confluence is achieved by first proving local confluence and then by applying Newman's Lemma.
\end{proof}

 Given $ s \in  \repsTerms(\mathcal{R})(\gamma; \ty),$ we denote by $\NF{s} $ its unique normal form. As a corollary of subject reduction, we get that $\NF{s} \in  \repsTerms(\mathcal{R})(\gamma ; \ty).$ We shall now present an inductive characterization of $ \torp$-normal terms for the case where  $\mathcal{R}$ is a \emph{discrete} signature.

\begin{definition}\label{chosnorm}
Consider the following set, inductively defined:\[     \NF{\mathsf{\repTerms} (\mathcal{R})} \ni s   ::= v  [\vec{x}_1 := x_1] \dots [\vec{x}_n := x_n]  \qquad          v ::=     \seqdots{v}{1}{k}  \mid x                  \]
where $k,n \in \mathbb{N}, \gamma \vdash p : o $ with $ o $ being an atomic type and  $\delta \vdash v : \ty  $ with $\delta$ being atomic.
\end{definition}

\begin{proposition}\label{repnormc}
A term $ s \in \repTerms (\mathcal{R}) $ is a normal form for $ \torp$ iff there exists $ s' \in    \NF{\repTerms (\mathcal{R})}  $ s.t. $s \streq s' . $
\end{proposition}
\begin{proof}
$ (\Rightarrow ) $ By induction on $s. $ If $ s = x $ then $s \in \NF{\repTerms (\mathcal{R})} . $ If $  \seqdots{s}{1}{k} $ then $s_i $ are normal form. Then we apply the IH and get $ s'_i \in \NF{\repTerms (\mathcal{R})} $  s.t. $s_i \streq s'_i  $. By definition $ s'_i = v_i [\vec{x}_{i, 1} := x_{i,1}] \dots [\vec{x}_{i, 1} := x_{i,n_i}].$ We then set $ s' = \seqdots{v}{1}{k} [\vec{x}_{1, 1} := x_{1,n_1}]  \dots [\vec{x}_{k, 1} := x_{k,n_k}] .$ If $ s = p[\vec{x} := q] $ we have that $ p  $ is a normal form and $ q$ is a $ \beta$-normal form. We reason by cases on $ q .$ If $ q$ does not have $ \eta$-redexes, we apply the IH and conclude in a way similar to the list case. If $ q$ has $ \eta$-redexes, since $ s$  is $ \beta$ normal we have that $ q \notin \mathsf{LT} . $ We can prove that $ q = x [\vec{x}_1 := q_1] \dots [\vec{x}_1 := q_n]$ with $q_i $ hereditarely of the same shape. Hence we conclude by pushing out all the substitutions from left to right.
\end{proof}


\begin{figure}[!t]
			\begin{align*} 
			  & \beta \ \text{\emph{Root step}:} \quad  s [x^{\ty_1}_1, \dots, x^{\ty_k}_k := \ctxl[\seqdots{t}{1}{k}]] \to_{\beta}  \ctxl[\subst{s}{x_1, \dots, x_k}{t_1, \dots, t_k}           ]. &
			\\
	&\eta \	\text{\emph{Root step}:} \quad	 s \to_{\eta}  \vec{x} [\vec{x}^{\tyl} := s]                              \qquad \text{where } \vec{x} \text{ fresh }, \ \gamma \vdash s : \tyl , \ s \notin \mathsf{LT}. & \\
& \text{\emph{Contextual extensions}:}  \qquad \begin{prooftree} \hypo{  s \to_\beta s'  }\infer1{  \ctx[s] \to_\beta \ctx[s']   } \end{prooftree} \qquad \begin{prooftree} \hypo{  s \to_\eta s'  }\infer1{  \ctxe[s] \to_\eta \ctxe[s']   } \end{prooftree} \quad ( {\torp} = {\to_{\beta} \cup \to_{\eta}}) .& \\
		 & \text{   \emph{Structural equivalence}:} \quad \ctx[s [\vec{x} := t]  ] \streq \ctx[s] [\vec{x} := t] \qquad \vec{x} \notin \fv{\ctx}    .
		\end{align*}	
\caption{Representable reduction relations and structural equivalence.} \hrulefill
	\label{fig:torp}
\end{figure}


 \subparagraph{Free Representable Multicategories}
  Let $ \mathcal{R}  $ be a representable signature. First, we define a  multicategory $ \freerm{\mathcal{R}} $ by setting 
$\Ob{\freerm{\mathcal{R}}} = \mathcal{R}_0$ and $ \freerm{\mathcal{R}} ( \gamma;  \ty) = \repTerms (\mathcal{R})( \gamma; \ty) {/} {\sim} $ where $ {\sim} = {(\streq \cup \repeq )} .$
 Composition is given by substitution, identities are given by variables.  The operation is well-defined on equivalence classes and satisfies associativity, identity axioms. We also have that if $s \sim s', $ then $\NF{s} \streq \NF{s'}. $ We denote by $\eta_{\mathcal{R}} : \mathcal{R} \to \overline{\freerm{\mathcal{R}}}$ the evident inclusion.


\begin{proposition}[Representability]\label{prop:repr}
   We have a bijection
   $  \freerm{\mathcal{R}}(\gamma,  \tens{\ty}{1}{k}, \delta ; \ty) \cong \freerm{\mathcal{R}}( \gamma, \ty_{1}, \dots,  \ty_{k}, \delta ; \ty  )                $ multinatural in $ \gamma, \delta $ and natural in $ \ty ,  $
  induced by the map  $         [s]  \mapsto [ \subst{s}{x}{\seqdots{x}{1}{k}} ]      . $
\end{proposition}
\begin{proof}

Naturality follows from basic properties of substitution. Inverses are given by the maps $ (-)[\vec{x} := x] :    \freerm{\mathcal{R}}( \gamma, \ty_{1}, \dots,  \ty_{k}, \delta ; \ty  )   \to   \freerm{\mathcal{R}}(\gamma,  \tens{\ty}{1}{k}, \delta ; \ty)  . $ \end{proof}

\begin{definition}\label{unirep}
 Let $ \mathcal{R}$ be a representable signature and $ \mathsf{S} $ be a representable multicategory. Let $ i : \mathcal{R} \to \overline{\mathsf{S}} $ be a map of representable signatures. We define a family of maps $ \mathsf{RT}(i)_{\gamma, \ty} : \repTerms(\mathcal{R}) (\gamma ; \ty)  \to \mathsf{S}(i(\gamma) ; i(\ty)) $ by induction as follows:
 
\scalebox{0.9}{\parbox{1.05\linewidth}{\[  \mathsf{RT}(i)_{\ty, \ty}(x) = id_{i(\ty)}  \qquad    \mathsf{RT}(i)_{\gamma_1, \dots, \gamma_k, \tens{\ty}{1}{k}}(\seqdots{s}{1}{k} ) = \bigotimes_{i =1}^{k}  \mathsf{RT}(i)_{\gamma_i, \ty_i}(s_i) \]
\[    \mathsf{RT}(i)_{\delta_1, \gamma, \delta_2, \ty}(s[x_1, \dots, x_k := t]) =  \mathsf{let} (\mathsf{RT}(i)_{\delta_1 , \ty_1, \dots, \ty_k, \delta_2, \ty}(s)) \circ \seq{id_{\delta_1},\mathsf{RT}(i)_{\gamma, \tens{\ty}{1}{k}}(t), id_{\delta_2}}  \]
 \[  \mathsf{RT}(i)_{\gamma_1, \dots, \gamma_n, \ty}(f (s_1, \dots, s_n)) =   i(f) \circ \seq{\mathsf{RT}(i)(s_1), \dots, \mathsf{RT}(i)(s_n)  } . \]}}
\end{definition}


\begin{theorem}[Free Construction]
Let $ \mathsf{S} $ be a a representable multicategory and $ i :  \mathcal{R} \to \overline{\mathsf{S}}  $ a map of representable signatures. There exists a unique representable functor $ i^{\ast} : \freerm{\mathcal{R}} \to \mathsf{S}$ such that $ i = \overline{i^\ast} \circ \eta_{\mathcal{R}} . $\end{theorem}
\begin{proof}
The functor is defined exploiting Definition \ref{unirep}.\end{proof}

\subparagraph{Coherence Result} We fix a \emph{discrete} representable signature $\mathcal{R} .$ We show that if $  s, t \in   \freerm{\mathcal{R}}(\gamma ; \ty)   $, then $ s = t . $ Our proof strongly relies on the characterization of normal forms given in Proposition \ref{repnormc}.
 
\begin{lemma}
Let $ \gamma, \gamma'  $ be atomic contexts. If there exists a type $ \ty $ and normal terms $ s, s' $ such that $s, s' \in \repTerms(\mathcal{R})(\gamma; \ty) $ then $ \gamma = \gamma'  $ and $ s \streq s' . $
\end{lemma} 
 
\begin{theorem}\label{cohrepal}
 Let $  s, s'$ be normal terms s.t. $ s, s'\in   \repTerms (\mathcal{R}) (\gamma; \ty) ,$ then $ s \streq s' . $
 \end{theorem}
 \begin{proof}
 By Proposition \ref{repnormc}, $ s \streq t = (v [\vec{x}_1 := x_1] \dots [\vec{x}_p := x_p]) $ and $ s' \streq t' =  (v' [\vec{y}_1 := x'_1] \dots [\vec{y}_p := x'_{p'}] ).$ We prove that $ t \streq t' $ by induction on $ p \in \mathbb{N} . $ If $ p = 0  $ then $ t $ is either a list or a variable. We proceed by cases. If $ t = x $ then $ \gamma = o $ and $ \ty = o $ for some atomic type $ o . $ By the former lemma we have that $ t \streq t' . $ If $ t = \seqdots{v}{1}{k} $ the result is again a corollary of the former lemma since, by Definition \ref{chosnorm}, $\gamma $ is atomic. If $ p = n + 1 $ then $ t = v [\vec{x}_1 := x_1] \dots [\vec{x}_{n + 1} := x_{n+1}] $ and, by definition of typing we have 
\scalebox{0.9}{\parbox{1.05\linewidth}{ \[               \begin{prooftree}     \hypo{x_{n+1}  : \tyl \vdash x_{n+1}  : \tyl  }\hypo{  \delta_1, \vec{x}_{n + 1} : \tyl, \delta_2   \vdash  v [\vec{x}_1 := x_1] \dots [\vec{x}_n := x_{n}] : \ty }\infer2{ \delta_1, x_{n+1}  : \tyl ,  \delta_2  \vdash s : \ty}    \end{prooftree}      \]}}

 with $ \gamma = \delta_1, x_{n + 1} : \tyl , \delta_2 .$ Since $ t' \in \NF{ \repTerms } (\gamma; \ty),  $ there exists $ i \in \mathbb{N} $ such that $ t' = v' [\vec{y}_1 := x'_1] \dots [\vec{y}_i := x'_{i}] \dots  [\vec{y}_p := x'_{p'}] $ and $ x'_{i} = x_{n+1} . $ By \emph{structural equivalence} we have that $ t' \streq v' [\vec{y}_1 := x'_1]  \dots  [\vec{y}_p := x'_{p'}] \dots [\vec{y}_i := x_{i}] .$ By definition of typing and by the hypothesis we have that 
 
 \scalebox{0.9}{\parbox{1.05\linewidth}{ \[               \begin{prooftree}     \hypo{x'_{i}  : \tyl  \vdash x'_{i}  : \tyl }\hypo{  \delta_1,\vec{x}_i : \tyl, \delta_2   \vdash  v' [\vec{y}_1 := x'_1] \dots [\vec{y}_{p'} := x'_{p'}] : \ty }\infer2{ \delta_1, x_{n+1}  : \tyl ,  \delta_2  \vdash s' : \ty}    \end{prooftree}      \]}}
  
  By IH we have that 
$ v [\vec{x}_1 := x_1] \dots [\vec{x}_p := x_{n}] \streq   v' [\vec{y}_1 := x'_1]  \dots  [\vec{y}_p := x'_{p'}]$ . We can then conclude that $ t \streq t' ,$ by structural equivalence.
 \end{proof}
 
 \begin{theorem}[Coherence for Representable Multicategories] \label{cohrep}
 Let $ [s], [t] \in  \freerm{\mathcal{R}}(\gamma; \ty).$ Then $ [s] = [t] .$
 \end{theorem}

\begin{theorem}[Coherence for Monoidal Categories]
All diagrams in the free monoidal category on a set commute.
\end{theorem}
\begin{proof}
Corollary of the former theorem and Theorem \ref{eqmonrep}, by noticing that $ \mathsf{mon}(\freerm{\mathcal{R}}) $ is the free monoidal category on the underlying set of $ \mathcal{R}. $
\end{proof}
 
 \section{A Resource Calculus for Symmetric Representable Multicategories}\label{sec:sr}The symmetric representable terms have exactly the same syntax and operational semantics as the representable ones. We first extend the type system in order to account for symmetries. We then study the free constructions establishing an appropriate coherence result.

 The typing is defined in Figure \ref{fig:reps-calc}. It is easy to see that the representable type system consists of a fragment of the symmetric one, where we just consider identity permutations. We write $ \gamma \vdash_{\mathsf{srep}} s : \ty $ when we need to specify that the type judgment refers to the symmetric representable type system. We set $  \repsTerms (\mathcal{R})( \ty_1, \dots, \ty_n ; \ty) = \{ s \mid x_1 : \ty_1 , \dots, x_n : \ty_n \vdash_{\mathsf{srep}} s : \ty \text{ for some } x_i \in \fv{s}. 	\} . $

  \begin{remark}
The role of permutations in the type system of Figure \ref{fig:reps-calc} deserves some commentary. Instead of having an independent permutation rule, variables in contexts can be permuted only when contexts have to be merged. In this way, the system is \emph{syntax directed}. The limitation to the choice of \emph{shuffle permutation} is needed to get uniqueness of type derivations for terms. Indeed, consider $ s = \seq{\seq{x,y}, z } . $ If we allow the choice of arbitrary permutations, we could build the following derivations:

\scalebox{0.8}{\parbox{1.05\linewidth}{\[   \begin{prooftree}  \hypo{x : \ty \vdash x : \ty \quad y : b \vdash y : b}\hypo{ \sigma    }\infer2{y : b, x : \ty \vdash \seq{x,y} : (\ty \otimes b)} \hypo{z : \ty \vdash z : \ty}\hypo{id}\infer3{y : b, x : \ty, z : \ty \vdash s :  ( (\ty \otimes b) \otimes \ty ) }   \end{prooftree}            \qquad    \begin{prooftree}  \hypo{x : \ty \vdash x : \ty \quad y : b \vdash y : b}\hypo{ id    }\infer2{ x : \ty, y : b \vdash \seq{x,y} : (\ty \otimes b)} \hypo{z : \ty \vdash z : \ty}\hypo{\sigma \oplus id}\infer3{y : b, x : \ty, z : \ty \vdash s :  ( (\ty \otimes b) \otimes \ty ) }  \end{prooftree} \]}}

where $ \sigma $ is the swap. Thanks to the shuffle limitation, only the one on the left is allowed.
\end{remark} 
  
 \begin{proposition}[Canonicity of Typing]\label{prop:symcano}
If $ \pi \triangleright \gamma \vdash s : \ty $ and $ \pi' \triangleright \gamma \vdash s : \ty' $ then $\ty = \ty' $ and $  \pi = \pi' .$
\end{proposition}
\begin{proof}
By induction on $s . $   In the cases where a merging of type contexts happens, such as the list case, we rely on the properties of shuffle permutations and on the fact that type contexts are repetitions-free. Hence, the action of permutations on contexts is always fixedpoint-free.
\end{proof}

 \begin{figure}[!t]
	\centering
	\begin{gather*}		
     \begin{prooftree} \hypo{\ty \in \mathcal{R}_0}\infer1{ x : \ty \vdash x : \ty }  \end{prooftree} \qquad 
\begin{prooftree}  \hypo{ \gamma_1 \vdash s_1 : \ty_1 \dots \gamma_k \vdash s_k : \ty_k  }\hypo{\sigma \in\mathsf{shu}(\gamma_1,\dots, \gamma_k)}  \infer2{ \act{(\gamma_1, \dots, \gamma_k)}{\sigma}  \vdash {\seqdots{s}{1}{k}}: \tens{\ty}{1}{k} }  \end{prooftree} \\[1em] \qquad  
 \begin{prooftree} 
 \hypo{\gamma \vdash s : \tens{\ty}{1}{k}  }
  \hypo{ \delta, x_1 : \ty_1, \dots, x_k : \ty_k, \delta' \vdash t : b  }\hypo{\sigma \in \mathsf{shu}(\delta,\gamma', \delta')}\infer3{ \act{(\delta, \gamma, \delta')}{\sigma} \vdash {t [x_1^{\ty_1}, \dots, x_k^{\ty_k} := s] }: b  }  \end{prooftree}
  \end{gather*}
  \caption{\small Symmetric Representable Type System on a signature $\mathcal{R}$. We omit the case $ f(\vec{s}). $}
  \hrulefill
	\label{fig:reps-calc} 
\end{figure} 
 

\begin{proposition}\label{permrule}
The following rule is admissible:
$ \begin{prooftree} \hypo{\gamma \vdash s : \ty}\hypo{\sigma \in S_k}\infer2{ \act{\gamma}{\sigma} \vdash s : \ty }  \end{prooftree}. $
\end{proposition}
\begin{proof}
Easy induction on the structure of $ s ,$ exploiting Lemma \ref{shucan}.  
\end{proof}

The reduction relation is the same as the representable one, that we know to be strongly normalizing and confluent. We also have preservation of typing under reduction and structural equivalence. Given $ s \in  \repsTerms(\mathcal{R})(\gamma; \ty),$ we denote by $\NF{s} $ its unique normal form. As a corollary of subject reduction, we get that $\NF{s} \in  \repsTerms(\mathcal{R})(\gamma ; \ty).$

\subparagraph{Free Symmetric Representable Multicategories}

We now characterize the free symmetric representable construction. Given a representable signature $\mathcal{R}, $ we define a multicategory by setting $\Ob{\freesrm{\mathcal{R}}} = \mathcal{R}_0 $ and $\freesrm{\mathcal{R}} ( \gamma;  \ty) = \repsTerms (\mathcal{R})( \gamma; \ty) {/} { ( \streq \cup \repeq )}.$ Composition is given by substitution, identities are given by variables. The operation is well-defined on equivalence classes and satisfies associativity, identity axioms. One can prove that $ \freesrm{\mathcal{R}} $ is representable, by repeating the argument given for Proposition \ref{prop:repr}. The proof that $ \freesrm{\mathcal{R}} $ is symmetric is a direct corollary of Proposition \ref{permrule}:

\begin{proposition}[Symmetry]\label{prop:sym}  We have that $ \mathcal{M}(\gamma, \ty_1, \dots, \ty_k; \ty) = \mathcal{M}(\gamma; \ty_{\sigma(1)}, \dots, \ty_{\sigma (k)}; \ty)                 . $
\end{proposition}

 \begin{example}
An interesting example of structural equivalence is the following. Let $  s =  \seq{}[-:= x][-:=y] \text{ and } s' = \seq{}[-:= y][-:=x] , $ with $ s, s' \in \repsTerms (\mathcal{R})( (),(); ()  ) .$ We have that $    \seq{}[-:= x][-:=y] \streq  \seq{}[-:= y][-:=x]  , $ with $ x : (), y : () \vdash s : ()$ and $ y : (), x : () \vdash s' : ().$ This is the way our syntax validates the fact that permutations of the unit type \emph{collapse} to the identity permutation, since $ s$ corresponds to the identity permutation, while $ s'$ to the swapping of $ x $ with $y $.  \end{example}

\begin{definition}\label{unireps}
 Let $ \mathcal{R}$ be a representable signature and $ \mathsf{S} $ be a symmetric representable multicategory. Let $ i : \mathcal{R} \to \overline{\mathsf{S}} $ be a map of representable signatures. We define a family of maps $ \mathsf{RT}(i)_{\gamma, \ty} : \repsTerms(\mathcal{R}) (\gamma ; \ty)  \to \mathsf{S}(i(\gamma) ; i(\ty)) $ by induction as follows: 
 
 \scalebox{0.9}{\parbox{1.05\linewidth}{\[  \mathsf{RT}(i)_{\ty, \ty}(x) = id_{i(\ty)}  \qquad    \mathsf{RT}(i)_{\act{(\gamma_1, \dots, \gamma_k)}{\sigma}, \seqdots{\ty}{1}{k}}(\seqdots{s}{1}{k} ) = \left(  \bigotimes_{j =1}^k \mathsf{RT}(i)_{\gamma_j, \ty_j}(s_j) \right) \circ \sigma \]
\[    \mathsf{RT}(i)_{\act{\delta_1, \gamma, \delta_2}{\sigma}, \ty}(s[ x_1^{\ty_1}, \dots, x_k^{\ty_k} := t]) =   {((\mathsf{RT}(i)_{\delta_1 , \ty_1, \dots, \ty_k, \ty}(s))^\ast \circ \seq{id_{\delta_1},\mathsf{RT}(i)_{\gamma, \tyl}(t), id_{\delta_2}} )}  \circ  {\sigma}. \]}}
\end{definition}

\begin{theorem}[Free Construction]
Let $ \mathsf{S} $ be a a symmetric representable multicategory and $ i :  \mathcal{R} \to \overline{\mathsf{S}}  $ a map of representable signatures. There exists a unique symmetric representable functor $ i^{\ast}$ such that  $  i = \overline{i^\ast} \circ \eta_{\mathcal{R}} . $ \end{theorem}

\subparagraph{Coherence Result} Fix a discrete signature $ \mathcal{R}.$ We shall prove that morphisms in $\freesrm{\mathcal{R}} $ can by characterized by means of appropriate permutations of their typing context. This will lead the following coherence result for symmetric representable multicategories: two morphisms in $ \freesrm{\mathcal{R}} $ are equal whenever their \emph{underlying permutations} are the same. 

We start by defining the \emph{strictification} of a representable type $ \stricti{\ty}$, by induction as follows: $\mathsf{strict}(o) = o , \ \mathsf{strict}(\tens{\ty}{1}{k}) = \stricti{\ty_1}, \dots, \stricti{\ty_k} . $ $\mathsf{strict}(\ty)$ is the list of atoms that appear in the type $ \ty $. We extend the strictification to contexts in the natural way. Let $s \in \NF{\repsTerms (\mathcal{R})}(\gamma, \ty) $ and $ \sigma \in \stabi{\stricti{\gamma}} . $ We define the \emph{right action} of $\sigma $ on $s $, $\ract{s}{\sigma} $ by induction as follows: \[      \ract{x}{id} = x \qquad \ract{\seqdots{s}{1}{k}}{\sigma \circ (\bigoplus_{i =1}^k \sigma_i)} =                                                \act{\seq{ \ract{s_1}{\sigma_1}, \dots, \ract{s_k}{\sigma_k}   }}{\sigma} \] \[      \ract{(s[\vec{x}_1:= x_1] \dots [\vec{x}_n:= x_n])}{\sigma} = (\ract{s}{\sigma})  [\sigma(\vec{x}_1):= x_1] \dots [\sigma (\vec{x}_n) := x_n]      \] where $\sigma (x_1, \dots, x_k) $ stands for the image of $ x_1, \dots, x_k$ along the permutation $\sigma . $  

\begin{theorem}
Let $ s \in \NF{\repsTerms (\mathcal{R})}(\gamma, \ty) .  $ There exists a unique $\sigma \in \stabi{(\stricti{\gamma}} $ and a unique non-symmetric representable normal term $t $ such that $ s = \ract{t}{\sigma} . $
\end{theorem}
\begin{proof}
By induction on $s , $ exploiting Proposition \ref{prop:symcano}. \end{proof}
 
 Let $ s \in \NF{\repsTerms (\mathcal{R})}(\gamma; \ty) . $ We denote by $ \mathsf{sym}(s) $ the unique permutation given by the former theorem. Given $ s \in \repsTerms{A}(\gamma; \ty) $ we set $\mathsf{sym}(s) = \mathsf{sym}(\NF{s}) .  $ This definition is clearly coherent with the quotient on terms performed in the free construction.
 
 \begin{theorem}
Let $ s, s' \in \NF{\repsTerms (\mathcal{R})}(\gamma; \ty) .$ If $ \mathsf{sym}(s) = \mathsf{sym}(s') $ then $ s \streq s' .   $
\end{theorem}

\begin{theorem}[Coherence]\label{cohsm}
Let $ [s], [s'] \in \freesrm{\mathcal{A}}(\gamma; \ty) . $ If $ \mathsf{sym}([s]) = \mathsf{sym}[s']  $ then $ [s] = [s'] . $ 
\end{theorem}

\begin{theorem}[Coherence for Symmetric Monoidal Categories] Two morphisms in the free symmetric monoidal categories are equal if their underlying permutations are equal. \end{theorem}
\begin{proof}
Corollary of Theorems \ref{eqmonrep} and \ref{cohsm}.
\end{proof}

\section{A Resource Calculus for Symmetric Closed Multicategories}\label{sec:sc}

We consider the case of symmetric closed multicategories, which is orthogonal to the representable structures we introduced in the previous sections. This calculus corresponds to the resource version of linear $ \lambda$-calculus, where we have unbiased $ k $-ary $\lambda $-abstraction and (linear) application. We begin by defining the terms and their typings, then proceed to introducing their operational semantics. We conclude by characterizing the free construction \textit{via} well-typed equivalence classes of terms.

\subparagraph{Symmetric Closed Resource Terms} Let $\mathcal{L} $ be a closed signature. The \emph{symmetric closed resource terms} on $\mathcal{L} $ are defined by the following inductive grammar:
\[     \rTerms (\mathcal{L}) \ni s ::= x \in \mathcal{V} \mid \la{\seq{x^{\ty_1}_1, \dots, x^{\ty_k}_k}} s \mid s \seqdots{s}{1}{k}  \mid f (s_1, \dots, s_k)  \]
for $ k \in \mathbb{N} $ and $ f \in \mathsf{arr}(\mathcal{L}), \ty_i \in \mathcal{L} . $ A term of the shape $s \seqdots{s}{1}{k} $ is called a \emph{($ k$-linear) application}. A term of the shape $ \la{\seq{ x_1, \dots, x_k} }  s$ is called a \emph{($ k$-linear) $ \lambda$-abstraction}. Variables under the scope of a $ \lambda $-abstraction are bound.  We define the following subset of terms  $\mathsf{AT} = \{   \ctxl[\la{\vec{x}} t] \mid \text{ for some substitution context } \ctxl \text{ and term } t .        \} . $ Typing and contexts with hole are defined in Figure \ref{fig:closed-calc}. Given a term $ \gamma \vdash s : \ty $, there exists a unique type derivation for it.


\begin{figure}[!t] \begin{gather*} 
			\begin{prooftree} \hypo{\ty \in \mathcal{L}_0}\infer1{ x : \ty \vdash x : \ty }  \end{prooftree} \qquad   \begin{prooftree} \hypo{ \gamma, x_1 : \ty_1, \dots, x_k : \ty_k \vdash s : b  }\infer1{\gamma \vdash \la{\seq{x_1^{\ty_1}, \dots, x_k^{\ty_k} }} s : \tens{\ty}{1}{k} \multimap b}  \end{prooftree}                \\[1em]  \begin{prooftree}  \hypo{ \gamma_0 \vdash s : \tens{\ty}{1}{k} \multimap b } \hypo{ \gamma_1 \vdash t_1 : \ty_1 \dots \gamma_k \vdash t_k : \ty_k } \hypo{ \sigma \in \mathsf{shu}(\gamma_0, \dots, \gamma_k) }\infer3{ \act{ (\gamma, \delta)}{\sigma} \vdash s \seqdots{t}{1}{k} : b  }   \end{prooftree}  
			  \end{gather*} \hrulefill
			  \begin{align*}
&	\ctx ::=  \ctx ::= \hole{\cdot} \mid  s\seq{s_1, \dots, \ctx, \dots, s_k}   \mid  \ctx \seq{s_1,  \dots, s_k}   \mid \la{\seq{x_1, \dots, x_k}} \ctx \mid f (s_1, \dots, \ctx, \dots, s_k). 
			\\
&	\ctxe ::= \hole{\cdot} \mid    s\seq{s_1, \dots, \ctxe, \dots, s_k}   \mid  \ctxe \seq{s_1,  \dots, s_k} \quad (\ctxe \neq \hole{\cdot})  \mid \la{\seq{x_1, \dots, x_k}} \ctxe  \mid f (s_1, \dots, \ctxe, \dots, s_k) .
			\end{align*}
			   \caption{\small Symmetric closed type system on a signature $\mathcal{L}$ and contexts with one hole. Types are the elements of $ \mathcal{L}_0 .$ We omit the case of $ f (\vec{s}) . $}\hrulefill
	
	\label{fig:closed-calc}
\end{figure}

\subparagraph{Terms under Reduction} The reduction relation is defined in Figure \ref{fig:toc}. 
\begin{remark}
The definition of the $\beta $-reduction follows the standard choices for resource calculi. The novel technicality is the restriction of the $\eta $-reduction, that is justified again by the goal of obtaining a strongly normalizing reduction. Indeed, $\eta $-reduction is again not normalizing. The situation recalls what happens in the standard $ \lambda$-calculus and we deal with it adapting to our framework the restrictions introduced in \cite{mints:closed,jay:eta}.
\end{remark}



\begin{figure}[!t]			\begin{align*} 		  
		&	\text{\emph{$\beta$ Root-Step}: }	  (\la{\seq{x_1^{\ty_1}, \dots, x_k^{\ty_k}}} s) \seqdots{t}{1}{k} \to_{\beta}  \subst{s}{x_1, \dots, x_k}{t_1, \dots, t_k}   .
			\\
			&  \text{\emph{$\eta$ Root-Step}: }      s \to_{\eta}     \la{\vec{x}^{\tyl}} (s \vec{x})       \qquad \text{where } \vec{x}  \text{ fresh }, \gamma \vdash s : \tyl \multimap \ty, s \notin \mathsf{AT}. \\
& \text{ \emph{Contextual extensions}: }  \quad \begin{prooftree} \hypo{  s \to_\beta s' }\infer1{ \ctx[s] \to_\beta \ctx[s'] } \end{prooftree} \qquad   \begin{prooftree} \hypo{ s \to_\eta s' }\infer1{ \ctxe[s] \to_\eta \ctxe[s'] } \end{prooftree} \quad ( \to_{\mathsf{sc}} = \to_{\beta} \cup \to_{\eta} ). 
	\end{align*}\caption{Symmetric closed reduction relations.}
\hrulefill
	\label{fig:toc}
\end{figure}

To study the rewriting, we adapt the method introduced for representable terms. We first prove that typing is preserved under reduction. Then, we introduce a measure that decreases under $\eta $. We define the \emph{size} of a type by induction: $\size{o} = 0, \size{\seqdots{\ty}{1}{k} \multimap \ty} = 1 + \sum \size{\ty_i} + \size{\ty}. $  Given $\gamma \vdash s  : \ty $ we define a set of typed subterms of $ s$: $\mathsf{EST} (s) = \{\delta \vdash p : 	\ty \mid p \in  \mathsf{ST}(s) \setminus \mathsf{AT}  \text{ s.t. } \ctxe[\delta \vdash p : \ty] = s  \text{ for some context } \ctxe  \} .$ We set $ \eta (s) =  \sum_{\delta \vdash p : \ty \in \mathsf{EST}(s)} \size{\ty}  .  $ The proof of strong normalization and confluence is completely symmetrical to the representable case. Given $ s \in  \rTerms(\mathcal{R})(\gamma; \ty),$ we denote by $\NF{s} $ its unique normal form. As a corollary of subject reduction, we get that $\NF{s} \in  \rTerms(\mathcal{R})(\gamma ; \ty).$

\subparagraph{Free Symmetric Closed Multicategories} Let $ \mathcal{L}  $ be a closed signature, we define a  multicategory $ \freescm{\mathcal{L}} $ by setting $ \Ob{\freescm{\mathcal{L}}} = \mathcal{L}_0 $ and $ \freescm{\mathcal{L}} ( \gamma;  \ty) = \rTerms (\mathcal{L})( \gamma; \ty) {/} { \sim } $ where $ {\sim} = { =_\mathsf{sc} }. $
Composition is given by substitut{ion, identities are given by variables. The operation is well-defined equivalence classes and satisfies associativity and identity axioms. We also have that if $s \sim s', $ then $\NF{s} = \NF{s'}. $ We denote by $\eta_{\mathcal{L}} : \mathcal{L} \to \overline{\freescm{\mathcal{L}} }$ the evident inclusion. One can prove that $ \freescm{\mathcal{R}} $ is symmetric, by repeating the argument given in the previous section. This multicategory is also \emph{closed}:
\begin{theorem}
We have a bijection $ \freescm{\mathcal{L}}( \gamma; \seqdots{\ty}{1}{k} \multimap \ty)                  \cong \freescm{\mathcal{L}}( \gamma, \ty_1, \dots,  \ty_k; \ty) $
natural in $\ty $ and multinatural in $\gamma , $ induced by the maps $[s] \mapsto [s\seqdots{x}{1}{k}] .$
\end{theorem}
\begin{proof}
Naturality derives from basic properties of substitution. Inverses are given by the maps $ [s] \mapsto [\la{\seqdots{x}{1}{k}} s] . $
\end{proof}

\begin{definition}\label{embclosed}
Let $ \mathsf{E} $ be a symmetric closed multicategory and let $ i : \mathcal{L} \to \overline{\mathsf{E}} $ be a map of closed signatures. We define a family of maps $ \mathsf{RT}_{\gamma, \ty} : \rTerms(\mathcal{L}) (\gamma, \ty)  \to \mathsf{E}(i(\gamma), i(\ty)) $ by induction as follows: 

\scalebox{0.9}{\parbox{1.05\linewidth}{\[  \mathsf{RT}_{\ty, \ty}(x) = 1_{i(\ty)}  \qquad  \mathsf{RT}_{\gamma, \tyl \multimap \ty} (\la{\vec{x}} s) = \lambda ( \mathsf{RT}_{\gamma, \tyl, \ty} (s) ) \] \[ \mathsf{RT}_{ (\gamma_0, \dots, \gamma_k), \ty}(s \seqdots{t}{1}{k}) = ( ev \circ   \seq{\mathsf{RT}_{\gamma_0, \seqdots{\ty}{1}{k} \multimap \ty} (s), \mathsf{RT}_{\gamma_1, \ty_1}(t_1), \dots, \mathsf{RT}_{\gamma_k, \ty_1}(t_k)} )  \cdot \sigma .  \]}}
\end{definition}

\begin{theorem}[Free Construction]
Let $ \mathsf{S} $ be a a symmetric closed multicategory and $ i :  \mathcal{L} \to \overline{\mathsf{S}}  $ a map of representable signatures. There exists a unique symmetric closed functor $ i^{\ast} : \freescm{\mathcal{L}}  \to   \mathsf{S} $ such that $ \overline{i^\ast} \circ \eta_{\mathcal{L}} = i . $
\end{theorem}

\begin{theorem}[Coherence]
Let $ [s], [s'] \in \freescm{\mathcal{R}}(\gamma; \ty) . $ Then $ [s] = [s'] $ iff $\NF{[s]} \streq \NF{[s']} .$
\end{theorem}

\section{A Resource Calculus for Autonomous Multicategories}\label{sec:aut} In this section we present our calculus for autonomous multicategories. These structures bring together representability, symmetry and closure. For this reason, the calculus we will present is a proper extension of the ones we introduced before. Again, we follow the same pattern of Sections \ref{sec:rep} and \ref{sec:sc}, first introducing the typing, then studying the operational semantics and finally characterizing the free constructions.

\subparagraph{Autonomous Terms} Let $\mathcal{A} $ be an autonomous signature. The \emph{autonomous resource terms} on $\mathcal{A} $ are defined by the following inductive grammar:
\[ \autTerms (\mathcal{A}) \ni  s,t ::= x \mid \la{\seq{x^{\ty_1}_1, \dots, x^{\ty_k}_k}} s \mid st \mid \seqdots{s}{1}{k}    \mid s[x^{\ty_1}_1, \dots, x^{\ty_k}_k := t] \mid f(s_1, \dots, s_k)  \]
for $ k \in \mathbb{N} $ and $ f \in \mathsf{arr}(\mathcal{A}), \ty_i \in \mathcal{A} . $ Variables under the scope of a $ \lambda $-abstraction and of a  substitution are bound. The typing is given in Figure \ref{fig:aut-calc}. The calculi introduced in the previous sections can be seen as subsystems of the autonomous one.


 

Given a subterm $ p  $ of $ s$ we write $\mathsf{ty}(p)_s $ for the type of $ p $ in the type derivation of $ s.$ The mapping is functional as corollary of the former proposition.

\begin{figure}[!t]
	\centering 		 	\scalebox{0.9}{\parbox{1.05\linewidth}{\begin{gather*} \qquad \qquad 
			\begin{prooftree} \hypo{\ty \in \mathcal{A}_0}\infer1{ x : \ty \vdash x : \ty }  \end{prooftree} \qquad 
\begin{prooftree}  \hypo{ \gamma_1 \vdash s_1 : \ty_1 \dots \gamma_k \vdash s_k : \ty_k  } \hypo{ \sigma \in \mathsf{shu}(\gamma_1, \dots, \gamma_k) } \infer2{\act{(\gamma_1, \dots, \gamma_k)}{\sigma} \vdash {\seqdots{s}{1}{k}}: \tens{\ty}{1}{k} }  \end{prooftree}   \\[1em]
 \begin{prooftree} \hypo{ \gamma, x_1 : \ty_1, \dots, x_k : \ty_k \vdash s : b  }\infer1{\gamma \vdash \la{\seq{x_1^{\ty_1}, \dots, x_k^{\ty_k}}} s : \tens{\ty}{1}{k} \multimap b}  \end{prooftree}                \qquad  \begin{prooftree}  \hypo{ \gamma \vdash s : \tyl \multimap b } \hypo{ \delta \vdash t : \tyl } \hypo{ \sigma \in \mathsf{shu}(\gamma, \delta) }\infer3{ \act{ (\gamma, \delta)}{\sigma} \vdash {st} : b  }   \end{prooftree}  \\[1em] \qquad \qquad \begin{prooftree} \hypo{\gamma \vdash s : \tens{\ty}{1}{k}  } \hypo{ \delta_1, x_1 : \ty_1, \dots, x_k  : \ty_k, \delta_2 \vdash t : b  }\hypo{ \sigma \in \mathsf{shu}(\gamma, \delta_1, \delta_2) }\infer3{ \act{(\gamma, \delta_1, \delta_2)}{\sigma} \vdash {t [x_1^{\ty_1}, \dots, x_k^{\ty_k} := s] }: b  }  \end{prooftree}
			  \end{gather*}
	\caption{\small Autonomous type system on a signature $\mathcal{A}$. We omit the case of $ f (\vec{s}) . $}\hrulefill
	\label{fig:aut-calc}}}
\end{figure}

\subparagraph{Terms under Reduction} The reduction relation $\toaut$, together with its subreductions $\beta$ and $ \eta  $  are defined by putting together the reductions $ \torp $ (Figure \ref{fig:torp}) and $\to_{\mathsf{sc}} $ (Figure \ref{fig:toc}). The same happens with structural equivalence. The reduction satisfies subject reduction, strong normalization and confluence. The proofs build on the results of the previous sections. As decreasing measures, we use the size of a term for $\beta $-reduction and the sum of the two $\eta $ measures we defined in the previous sections for $\eta $-reduction.

\subparagraph{Free Autonomous Multicategories} Let $ \mathcal{A}  $ be an autonomous signature, we define a  multicategory $ \freeaut{\mathcal{A}} $ by setting $ \Ob{\freeaut{\mathcal{A}}} = \mathcal{A}_0 $ and $                
 \freeaut{\mathcal{A}} ( \gamma;  \ty) = \autTerms (\mathcal{A})( \gamma; \ty) {/} { \sim} $ where $ \sim $ is the equivalence $\streq \cup =_\mathsf{aut} . $ Composition is given by substitution, identities are given by variables. The operation is well-defined on equivalence classes and satisfies associativity and identity axioms. We also have that if $ s \sim s' $ then $ \NF{s} \streq \NF{s'} . $ We denote by $ \eta_{\mathcal{A}} : \mathcal{A} \to \overline{\freeaut{\mathcal{A}}}  $ the evident inclusion. One can prove that this multicategory is symmetric, representable and closed by importing the proofs given in the previous sections.

\begin{definition}\label{embaut}
Let $ \mathsf{S} $ be an autonomous multicategory and let $ i : \mathcal{A} \to \overline{\mathsf{S}} $ be a map of autonomous signatures. We define a family of maps $ \mathsf{RT}_{\gamma, \ty} : \autTerms(\mathcal{A}) (\gamma, \ty)  \to \mathsf{E}(i(\gamma), i(\ty)) $ by induction, extending Definitions \ref{unireps} and \ref{embclosed} in the natural way.
\end{definition}

\begin{theorem}[Free Construction]
Let $ \mathsf{S} $ be a an autonomous multicategory and $ i :  \mathcal{A} \to \overline{\mathsf{S}}  $ a map of autonomous signatures. There exists a unique autonomous functor $ i^{\ast} : \freeaut{\mathcal{A}}  \to   \mathsf{S} $ such that $ \overline{i^\ast} \circ \eta_{\mathcal{A}} = i . $
\end{theorem}

\begin{theorem}[Coherence]
Let $ [s], [s'] \in \freeaut{\mathcal{R}}(\gamma; \ty) . $ Then $ [s] = [s'] $ iff $\NF{[s]} \streq \NF{[s']} .$
\end{theorem}

\section{Conclusion}\label{sec:conc} We established a formal correspondence between resource calculi and appropriate linear multicategories, providing coherence theorems by means of normalization. As future work, we consider two possible perspectives. It is tempting to parameterize our construction over the choice of allowed \emph{structural rules} on typing contexts. For instance, while the choice of permutations (\textit{i.e.}, symmetries) gives \emph{linear} structures, the choice of arbitrary functions between indexes would give \emph{cartesian} structures. In this way, we would achieve a general method to produce type theories for appropriate \emph{algebraic theories}, in the sense of \cite{hyland:alg}. Another perspective is the passage to the \emph{second dimension}, following the path of \cite{fiore:tt}. In this way, the rewriting of terms would become visible in the multicategorical structure itself. Coherence by normalization could then be upgraded to a method of \emph{coherence by standardization}, exploiting a rewriting relation on reduction paths.

\printbibliography

\appendix

\section{Appendix}

We detail some technical proofs. We use $ ::$ to denote list concatenation. If we have $ \gamma_1 \vdash t_1 : \ty_1, \dots, \gamma_k \vdash t_k ; \ty_k  $ we will abuse the language and often abbreviate with $ \vec{\gamma} \vdash \vec{t} : \tyl .  $  We denote by $ \abs{\gamma} $ the set of variables appearing in $\gamma. $ We denote by $\gamma(x) $ the index of $x $ in $ \gamma $ given by the linear order.

\section{Multicategories and Signatures}

\subparagraph{Multicategories}
\begin{definition}
A \emph{morphism} of multigraphs $ \mathcal{F} : \mathcal{G} \to \mathcal{H}  $ is the collection of the following data:
\begin{itemize}
\item A function $ \mathcal{F}_0 : \mathcal{G}_0 \to \mathcal{H}_0 .$
\item For every $ a_1, \dots, a_n, b \in \mathcal{G}_0,  $ a family of maps
 \[        \mathcal{F}_{\ty_1, \dots, \ty_n, b} : \mathcal{G}(\ty_1, \dots, \ty_n, b) \to \mathcal{H}(\mathcal{F}_0 (\ty_1) , \dots, \mathcal{F}_0(\ty_n),  \mathcal{F}_0 (b))                \]
\end{itemize}
\end{definition}

We shall just wrote $ \mathcal{F}(A), \mathcal{F}(s) $ for the action of a morphisms over nodes and edges. We shall write $ s : \ty_1, \dots, \ty_n \to b , \source{s} = \ty_1, \dots, \ty_n, \targ{s} = b$ meaning that $ s \in G(\ty_1, \dots, \ty_n, b) . $ 

\begin{definition} A \emph{functor} of multicategories $ \mathsf{F} : \mcati \to \mcati ' $ consists of a morphisms of multigraphs that preserves composition and identities.
\begin{enumerate}
\item A \emph{representable functor} between representable multicategories is a functor $ \mathsf{F} : \mcati \to \mcati ' $ s.t. $  \mathsf{F}(\ty_1 \otimes_{\mcati} \dots \otimes_{\mcati} \ty_k) = (\mathsf{F}(\ty_1) \otimes_{\mcati '} \dots \otimes_{\mcati '} \mathsf{F}(\ty_k)) $ and $ \mathsf{F}(\mathsf{re}^{\mcati}_{\ty_1, \dots, \ty_k}) =  \mathsf{re}^{\mcati '}_{\mathsf{F} (\ty_1), \dots, \mathsf{F}(\ty_k) } .   $
\item A \emph{symmetric} functor between symmetric multicategories is a functor $   \mathsf{F} : \mcati \to \mcati ' $ s.t. $ \mathsf{F}(s \cdot \sigma) = \mathsf{F}(s) \cdot \sigma . $
\item A \emph{closed} functor between closed multicategories is a functor  $ \mathsf{F} : \mcati \to \mcati ' $ s.t. $ \mathsf{F}((\ty_1 \otimes_{\mcati} \dots \otimes_{\mcati} \ty_k) \multimap_{\mcati} \ty) = (\mathsf{F}(\ty_1) \otimes_{\mcati '} \dots \otimes_{\mcati '} \mathsf{F}(\ty_k)) \multimap_{\mcati '} \mathsf{F}(\ty) $ and $ \mathsf{F}(ev_{\ty_1, \dots, \ty_k, \ty}) = ev_{\mathsf{F}(\ty_1), \dots, \mathsf{F}(\ty_k), \mathsf{F}(\ty) }.$
\end{enumerate} \end{definition}

\subparagraph{Monoidal vs Representable}

We discuss an adjunction between the 2-category $ \MON $ of monoidal categories, strong monoidal functors and monoidal natural transformations and the 2-category $ \REP $ of representable multicategories, representable functors and representable multinatural transformations. We shall also show how this adjunction lift to the symmetric case.

Given a monoidal category $ (\mocati, \otimes, I)  $ we shall define a representable multicategory $ \mathsf{rep}(\mocati) \in \REP  $  by exploiting its monoidal structure. Given $ \ty_1. \dots, \ty_k \in \mocati , $ let $ \tens{\ty}{1}{k} := ( \dots  (\ty_1 \otimes \ty_2 ) \otimes \dots ) \otimes \ty_k    $ for $ k > 0 $ and $ \tens{\ty}{1}{k} := I $ for $ k = 0 . $  
Then we fix $  \mathsf{rep}(\mocati)(\ty_1, \dots, \ty_n, \ty) = \mocati (\tens{\ty}{1}{n}, \ty) .    $ composition is defined as follows: 
given $ f_1 \in \mathsf{rep}(\mocati)( \gamma_1, b_1   ), \dots,   f_1 \in \mathsf{rep}(\mocati)( \gamma_n, b_n   ) $ and $ f  \in \mathsf{rep}(\mocati)( b_1, \dots, b_n,  \ty   ), $ we set 
\[        f \circ \seqdots{f}{1}{n}  := f \circ (( f_1 \otimes \cdots \otimes f_n  )  \circ \alpha  )                 \]
where $ \alpha $ is an appropriate choice of isomorphism built out of the associators in $ \mocati . $  The former construction can be extended to the symmetric case in the natural way, by defining symmetries of   $\mathsf{symon}(\mocati)$ \textit{via} the right action of permutations on representable maps, \textit{i.e.} $\sigma_{\ty, b} :b \otimes \ty \to \ty \otimes b := \act{\mathsf{let}(\mathsf{re}_{\ty,b} )}{\sigma}.$

Given a multicategory $ \mcati ,$ we shall define a monoidal category $\mathsf{mon}(\mcati) $ as follows. We set $ \Ob{\mathsf{mon}(\mcati)} = \Ob{\mcati} $ and $ \mathsf{mon}(\mcati)(\ty, b) = \mcati (\ty,b) . $ Composition and identities are then inherited. The tensor product is given by the binary tensors of $ \mcati ,$ the unit is the 0-ary tensor $() . $ The tensor on morphisms is defined as $ f \otimes g = \mathsf{re} \circ \seq{f,g} . $ The associators $  \alpha_{\ty,b, c} : (\ty \otimes b) \otimes c \cong \ty \otimes (b \otimes c)  $ are given by the maps $ \mathsf{let}_{(\ty \otimes b), c} (\mathsf{let}_{\ty,b}(\mathsf{re}_{\ty, (b \otimes c)})) .$

\subparagraph{Signatures} Let $\seq{\atm, \mathcal{N} }$ be either a representable or a closed signature. We denote by $ \mathsf{aut}(\mathcal{N}) $ the autonomous signature generated from $\mathcal{N}  ,$ whose node are freely  generated on $\atm $ as follows: 
\[  \mathsf{aut}(\mathcal{N})_0 \ni \ty ::= o \in \atm \mid \tens{\ty}{1}{k} \mid \tens{\ty}{1}{k} \multimap \ty \quad ( k \in \mathbb{N}) .\] 
We remark that we have a structure-preserving injective function  $ \iota : \mathcal{N}_0 \hookrightarrow \mathsf{aut}(\mathcal{N})_0 .$ Then we set \[   \mathsf{aut}(\mathcal{N})(\gamma, \ty) = \begin{cases} \mathcal{N}(\gamma', \ty' ) & \text{ if } \gamma = \iota(\gamma'), \quad \ty = \iota (\ty'); \\
\emptyset & \text{ otherwise.}   \end{cases}            \]
 We get and evident structure-preserving embedding $  \mathcal{N} \hookrightarrow \mathsf{aut}(\mathcal{N}) .$

\section{ (Symmetric) Representable Resource Terms}

\subsection{Representable Case}

\subparagraph{Substitution and Reduction}

\begin{definition}
Let $ s \in \repTerms $ and $   x_1,\ dots, x_k \in \mathcal{V} . $ We define the list of \emph{occurrences} of $x_1, \dots, x_ks $ in $s $ by induction as follows: 
\[   \mathsf{occ}_{\vec{x}} (x) = \begin{cases}    x & \text{ if } x\in \vec{x} ; \\ \seq{} & \text{ otherwise.}       \end{cases}      \qquad  \mathsf{occ}_{\vec{x}}(\seqdots{s}{1}{k}) = \bigoplus \mathsf{occ}_{\vec{x}} (s_i )    \] \[\mathsf{occ}_{\vec{x}}(s[\vec{y} := u]) = \mathsf{occ}_{\vec{x}} (s) :: \mathsf{occ}_{\vec{x}}(u) .  \]
\end{definition}

From now on, whenever we deal with linear substitutions we shall always assume that they are well-defined, hence the substituted variables are contained in the free variables of the considered term.

\begin{definition}[Representable Linear Substitution]
Let $  s,  t_1, \dots, t_k \in \repTerms $ and $ \  x_1, \dots, x_k  \subseteq \fv{s}. $ We define the \emph{($k $-ary linear) substitution} of $ x_1, \dots, x_k$ by $ t_1, \dots, t_k $ in $s $ by induction as follows:
\[     \subst{x}{x}{t} = t \qquad \subst{x}{}{} = x \quad (\text{ if } k = 0) \] \[     \subst{\seqdots{s}{1}{n}}{\vec{x}}{\vec{t}} =  \seq{\subst{s_1}{\vec{x}_1}{\vec{t}_1}, \dots, \subst{s_n}{\vec{x}_n}{\vec{t}_n}} \quad (\vec{x} = \bigoplus \vec{x}_i, \vec{t} = \bigoplus \vec{t}_i \text{ and } \vec{x}_i = \mathsf{occ}_{\vec{x}}(s_i) ) \]
\[   \subst{s_1[\vec{y} := s_2 ]}{\vec{x}}{\vec{t}} = \subst{s_1}{\vec{x}_1}{\vec{u}_1}[\vec{y} := \subst{s_2}{\vec{x}_2}{\vec{t}_2}]  \quad (\vec{x} = \vec{x}_1 :: \vec{x}_2, \vec{t} = \vec{t}_1 :: \vec{t}_2 \text{ and } \vec{x}_i = \mathsf{occ}_{\vec{x}}(s_i) .  \]
\end{definition}

We characterize $\beta $-normal forms as follows:
\[          \NF{\repTerms}_\beta   \ni s ::=   \seqdots{s}{1}{k}[\vec{x}_1 := v_1] \dots  [\vec{x}_n := v_n] \mid v \qquad v := x [x := v ]        \]

\begin{lemma}
A term $ s$ is a $ \beta$-normal form iff $s \in  \NF{\repTerms}_\beta .$ 
\end{lemma}
\begin{proof}
By induction on $ s.$  If $s = \seqdots{s}{1}{k} $ we have that $ s_i$ are $\beta $-normal, we conclude by applying the IH. If $ s = p[\vec{x}:=q]$ then $ p$ and $ q$ are  $ \beta$-normal. By IH $ p, q \in \NF{\repTerms}_\beta .$ We reason by cases on $ q .$ The case $ p = \seqdots{s}{1}{k}[\vec{x}_1 := v_1] \dots  [\vec{x}_n := v_n]$ is not possible, otherwise we would have a $ \beta$-redex.   
\end{proof}

We consider another inductive set of normal forms:
\[          \NF{\repTerms}'_\beta   \ni s ::=    p[\vec{x}_1 := x_1] \dots  [\vec{x}_n := x_n]   \mid p ::= \seqdots{p}{1}{k} \mid x  \]

\begin{proposition}
If $ s   \in  \NF{\repTerms}_\beta$ there exists $ t \in  \NF{\repTerms}'_\beta $ s.t. $ s \streq t .$
\end{proposition}
\begin{proof}
By induction on $ s .$ If $  \seqdots{s}{1}{k}[\vec{x}_1 := v_1] \dots [\vec{x}_n := v_n] $ then $s_i $ are $\beta$-normal form. Then we apply the IH and get $ s'_i \in \NF{\repTerms (\mathcal{R})}' $  s.t. $s_i \streq s'_i  $. By definition $ s'_i = p_i [\vec{x}_{i, 1} := x_{i,1}] \dots [\vec{x}_{i, 1} := x_{i,n_i}].$ We then set $ s' = \seqdots{p}{1}{k} [\vec{x}_{1, 1} := x_{1,n_1}]  \dots [\vec{x}_{k, 1} := x_{k,n_k}] [\vec{x}_1 := v_1] \dots [\vec{x}_n := v_n] .$ If $ s = x[x :=v] $ by IH we have $ v \streq t$ with $ t' = p [\vec{x}_1 := x_1 ] \dots [\vec{x}_n := x_n ] $ with $ p = y$, otherwise we would have a  $\beta $-redex. Then we set $ t  = x[x:= y ] [\vec{x}_1 := x_1 ] \dots [\vec{x}_n := x_n ] . $ 
\end{proof}

\begin{proposition}
A term $ s \in \repTerms (\mathcal{R}) $ is a normal form for $ \torp$ iff there exists $ s' \in    \NF{\repTerms (\mathcal{R})}  $ s.t. $s \streq s' . $
\end{proposition}
\begin{proof}
Corollary of the former proposition.
\end{proof}

\subparagraph{Free Construction}

\begin{definition}\label{rep:unirep}
 Let $ \mathcal{R}$ be a representable signature and $ \mathsf{S} $ be a representable multicategory. Let $ i : \mathcal{R} \to \overline{\mathsf{S}} $ be a map of representable signatures. We define a family of maps $ \mathsf{RT}(i)_{\gamma, \ty} : \repTerms(\mathcal{R}) (\gamma ; \ty)  \to \mathsf{S}(i(\gamma) ; i(\ty)) $ by induction as follows:
 
\scalebox{0.9}{\parbox{1.05\linewidth}{\[  \mathsf{RT}(i)_{\ty, \ty}(x) = id_{i(\ty)}  \qquad    \mathsf{RT}(i)_{\gamma_1, \dots, \gamma_k, \tens{\ty}{1}{k}}(\seqdots{s}{1}{k} ) = \bigotimes_{i =1}^{k}  \mathsf{RT}(i)_{\gamma_i, \ty_i}(s_i) \]
\[    \mathsf{RT}(i)_{\delta_1, \gamma, \delta_2, \ty}(s[x_1, \dots, x_k := t]) =  \mathsf{let} (\mathsf{RT}(i)_{\delta_1 , \ty_1, \dots, \ty_k, \delta_2, \ty}(s)) \circ \seq{id_{\delta_1},\mathsf{RT}(i)_{\gamma, \tens{\ty}{1}{k}}(t), id_{\delta_2}}  \]
 \[  \mathsf{RT}(i)_{\gamma_1, \dots, \gamma_n, \ty}(f (s_1, \dots, s_n)) =   i(f) \circ \seq{\mathsf{RT}(i)(s_1), \dots, \mathsf{RT}(i)(s_n)  } . \]}}
\end{definition}

We recall that in a representable category, $ f \otimes g $, with $\gamma \vdash f : \ty, \delta \vdash g : b $ is defined \textit{via} the maps $ \mathsf{re}_{a,b} : a,b \to (a \otimes b)  $ as $\mathsf{re}_{a,b} \circ \seq{f,g} .  $ We have that $ \mathsf{re}_{a,b} \circ \seq{f \circ f',g \circ g'} = (\mathsf{re}_{a,b} \circ \seq{f ,g }) \circ \seq{f',g'}$ by associativity of composition. We shall constantly abbreviate $  \seq{\mathsf{RT}_{\gamma_i, \ty_i}(t_i)}_{i \in \length{\gamma}}$ as $\mathsf{RT}_{ \vec{\gamma}, \tyl}(\vec{t}) . $

\begin{proposition}\label{functrep}The following statements hold.\begin{enumerate}\item $ \mathsf{RT}_{\vec{\gamma}, \ty}(\subst{s}{\vec{x}}{\vec{t}}) = \mathsf{RT}_{\gamma, \ty}(s) \circ \seq{\mathsf{RT}_{\gamma_i, \ty_i}(t_i)}_{i \in \length{\gamma}} .$\item If $ s \torp s' $ then  $  \mathsf{RT}_{\gamma, \ty}(s) = \mathsf{RT}_{\gamma, \ty}(s') .$\item If $s \repeq s' \text{ or } s \streq s' $ then $  \mathsf{RT}_{\gamma, \ty}(s) = \mathsf{RT}_{\gamma, \ty}(s') .$
\item If $ s \streq  s'$ then $ \mathsf{RT}_{\gamma, \ty}(s) = \mathsf{RT}_{\gamma, \ty}(s') .$
\end{enumerate}\end{proposition}
\begin{proof}
\begin{enumerate}
\item By induction on $ s .$ If $ s = x $ then $ \vec{t} = t$ and $ \subst{s}{\vec{x}}{\vec{t}} = t .$ Then we have $ \mathsf{RT}_{\vec{\gamma}, \ty}(\subst{s}{\vec{x}}{\vec{t}}) = \mathsf{RT}_{\gamma, \ty}(t) = id \circ \mathsf{RT}_{\gamma, \ty}(t) . $ If $ s = \seqdots{s}{1}{k} $ with $ \delta_i^1, \vec{x}_i, \delta_i^2 \vdash s_i : \ty_i $ with $\vec{x} = \bigoplus \vec{x}_i . $ Then $ \subst{s}{\vec{x}}{\vec{t}} = \seq{\subst{s_1}{\vec{x}_1}{\vec{t}_1}, \dots, \subst{s_k}{\vec{x}_k}{\vec{t}_k}}. $ By IH we have that $ \mathsf{RT}_{\delta_i^1, \vec{\gamma}, \delta_i^2, \ty}(\subst{s_i}{\vec{x}_i}{\vec{t}_i}) =  \mathsf{RT}_{\delta_i^1, \tyl_i, \delta_i^2, \ty}(s_i) \circ \mathsf{RT}_{ \vec{\gamma}_i, \tyl}(\vec{t}_i) .$ By associativity of composition, we have \[     \bigotimes   \left(  \mathsf{RT}_{\delta_i^1, \tyl_i, \delta_i^2, \ty}(s_i) \circ \mathsf{RT}_{ \vec{\gamma}_i, \tyl}(\vec{t}_i) \right) =          \left( \bigotimes     \mathsf{RT}_{\delta_i^1, \tyl_i, \delta_i^2, \ty}(s_i) \right) \circ \mathsf{RT}_{ \vec{\gamma}, \tyl}(\vec{t}) .  \] We can then conclude.  If $ s = p [\vec{y} :=q] $ then $ \subst{s}{\vec{x}}{\vec{t}} =\subst{p}{\vec{x}_1}{\vec{t}_1}[\vec{y} := \subst{q}{\vec{x}_2}{\vec{t}_2}]  $ with $ \delta_1, \vec{x}_1 : \tyl_1,  \vec{y} : \vec{b}, \delta_3 \vdash p : \ty $  and $   \vec{x}_2 : \tyl_2, \zeta \vdash q : \vec{b} .  $
\item By induction on $ s \torp s' . $ If it is a $\beta $ step, then it's a direct corollary of the former point of this lemma. If $ s \torp \vec{x} [\vec{x} := s], $ we have that $\mathsf{RT}_{\gamma, \ty}(\vec{x} [\vec{x} := s]) =  (\mathsf{let}( \mathsf{re} )) \circ s .   $ We observe that $ (\mathsf{let}( \mathsf{re} ))$ is the identity morphism; we can then conclude. The contextual cases are a direct application of the IH.
\item If $ s \repeq s' $, we observe that there exists $ n \in \mathbb{N} $ s.t. $ s = s_0 \leftrightarrow  s_1 \leftrightarrow \dots \leftrightarrow s_n = s' $ where $ \leftrightarrow $ stands for the reflexive and symmetric closure of $ \torp .$ We prove the result by induction on $ n .$ If $ n = 0 $ the result is a direct corollary of the former point of the lemma. If $ n = p + 1 $ the result is a direct corollary of the IH.
\item If $ s \streq s', $ the result is by induction on the judgment $s \streq s'. $  If $  \seq{s_1, \dots, s_i[\vec{x}:= t], \dots, s_k} \streq  \seq{s_1, \dots, s_i[, \dots, s_k}\vec{x}:= t] $ we have that \[\mathsf{RT}_{\gamma_1, \dots, \gamma_k, \seqdots{\ty}{1}{k}}(\seq{s_1, \dots, s_i[\vec{x}:= t], \dots, s_k}) = \] \[  \mathsf{RT}_{\gamma_1, \ty_1}(s_1) \otimes \dots  \mathsf{RT}_{\gamma_i, \ty_i}(s_i [\vec{x} := t]) \otimes \dots \otimes \mathsf{RT}_{\gamma_k, \ty_k}(s_k) \]
then we can conclude by associativity of the composition. A similar argument can be applied to the other ground steps for structural equivalence. The contextual cases are a direct application of the IH.
\end{enumerate}
\end{proof}

\begin{theorem}[Free Construction]
Let $ \mathsf{S} $ be a a representable multicategory and $ i :  \mathcal{R} \to \overline{\mathsf{S}}  $ a map of representable signatures. There exists a unique representable functor $ i^{\ast} : \freerm{\mathcal{R}} \to \mathsf{S}$ such that $ i = \overline{i^\ast} \circ \eta_{\mathcal{R}} . $\end{theorem}
\begin{proof}
The functor is defined exploiting Definition \ref{unirep}.  By the former proposition, $ \iota^\star$ is well-defined and preserves composition and the identities. Given another functor $ i' : \freerm{\mathcal{R}} \to \mathsf{S} $ s.t.   $i = \overline{i'} \circ \eta_{\mathcal{R}} ,$ one proves that $ i^\ast = i' $ pointwise, exploiting the fact that $i' $ has to preserve the representable structure.\end{proof}

\subparagraph{Coherence}

\begin{lemma}
Let $\seq{\atm, \mathcal{R}} $ be a discrete signature. We have that $ \mathsf{mon}(\freerm{\mathcal{R}})$ is the free monoidal category on $ \atm. $
\end{lemma}
\begin{proof}
the map $ \atm \to \Ob{\mathsf{mon}(\freerm{\mathcal{R}})}$ is trivial. Given a monoidal category $ \mocati ,$ the functor $ i^\ast : \mathsf{mon}(\freerm{\mathcal{R}} \to \mcati  $ is build out of Definition \ref{rep:unirep}. Uniqueness is proved pointwise. \end{proof}

\begin{lemma}
Let $ \gamma, \gamma'  $ be atomic contexts. If there exists a type $ \ty $ and normal terms $ s, s' $ such that $s, s' \in \NF{\repTerms(\mathcal{R})}(\gamma; \ty) $ then $ \gamma = \gamma'  $ and $ s = s' . $
\end{lemma} 
\begin{proof}by induction of $ \ty . $ We remark that the terms $s.s' $ cannot be substitution, since the contexts are atomic. If $ \ty = o , $ since $s, s' $ are normal, then $ s = x = s' . $ If $ \ty =  \seqdots{\ty}{1}{k}$, since $ \gamma , \gamma'$ are atomic, by definition  $ s = \seqdots{s}{1}{k}, s' = \seqdots{s'}{1}{k} $. By definition, we have $ \gamma_i \vdash s_i : \ty_i $ and $ \gamma'_i \vdash s'_i : \ty_i $ for $ i \in [k] . $ Moreover, $ \gamma = \gamma_1, \dots, \gamma_k $ and $ \gamma' = \gamma'_1, \dots, \gamma'_k . $ By IH, $ \gamma_i = \gamma'_i $ and $s_i = s'_i . $ We can then conclude. \end{proof}

\subsection{Symmetric Case}

 \subparagraph{Type System}

\begin{proposition}[Canonicity of Typing]\label{app:canonsym}If $ \pi \triangleright \gamma \vdash s : \ty $ and $ \pi' \triangleright\gamma \vdash s : \ty' $ then $\ty = \ty' $ and $  \pi = \pi' .$\end{proposition}\begin{proof} Corollary of Propositions \ref{emb} and \ref{aut:canonsym}.\end{proof}

\subparagraph{Coherence}

\begin{theorem}
Let $ s \in \NF{\repsTerms (\mathcal{R})}(\gamma, \ty) .  $ There exists a unique $\sigma \in \stabi{(\stricti{\gamma}} $ and a unique non-symmetric representable normal term $t $ such that $ s = \ract{t}{\sigma} . $
\end{theorem}
\begin{proof}
If $s = x $ the result is immediate. If $ s = \seqdots{s}{1}{k} $ with $ \gamma = \act{(\gamma_1, \dots, \gamma_k)}{\sigma} \vdash \seqdots{s}{1}{k} : \tens{\ty}{1}{k} $ and $\gamma  $ being atomic, by IH we have unique $ \sigma_1, \dots, \sigma_k \in  \mathsf{St}(\stricti{\gamma_i})$ and $ t_1, \dots, t_k \in \NF{\repTerms(A)} $ s.t. $ s_i = \ract{t_i}{\sigma_i}  $ for $ i \in [k] . $ Then, by definition, $ s = \ract{\seqdots{t}{1}{k}}{\sigma \circ (\sigma_1 \otimes \dots \otimes \sigma_k)} .$ Uniqueness derives by Proposition \ref{app:canonsym}. If $ s = p [\vec{x}_1 := x_1] \dots [\vec{x}_n := x_n], $ By IH there exists unique $ \sigma  $ and $ t$ s.t. $ p = \ract{t}{\sigma} . $ Then we can conclude by the fact that the action of $\sigma $ on variables is fixedpoint-free. 
 \ref{prop:symcano}. \end{proof}

\section{Symmetric Closed Case}

\subparagraph{Substitution and Reduction}

\begin{definition}Let $ s \in \rTerms $ and $   x_1,\ \dots, x_k \in \mathcal{V} . $ We define the list of \emph{occurrences} of $x_1, \dots, x_ks $ in $s $ by induction as follows: 
\[   \mathsf{occ}_{\vec{x}} (x) = \begin{cases}    x & \text{ if } x\in \vec{x} ; \\ \seq{} & \text{ otherwise.}       \end{cases}      \qquad  \mathsf{occ}_{\vec{x}}(s\seqdots{s}{1}{k}) = \mathsf{occ}_{\vec{x}}(s) \oplus \bigoplus \mathsf{occ}_{\vec{x}} (s_i )    \] \[\mathsf{occ}_{\vec{x}}( \la{\vec{y} } s) = \mathsf{occ}_{\vec{x}} (s) .  \]
\end{definition}

\begin{definition}[Linear Substitution]
Let $  s,  t_1, \dots, t_k \in \repTerms $ and $ \  x_1, \dots, x_k  \subseteq \fv{s}. $ We define the \emph{($k $-ary linear) substitution} of $ x_1, \dots, x_k$ by $ t_1, \dots, t_k $ in $s $ by induction as follows:
\[     \subst{x}{x}{t} = t \qquad \subst{x}{}{} = x \quad (\text{ if } k = 0) \] \[     \subst{s_0\seqdots{s}{1}{n}}{\vec{x}}{\vec{t}} =  \subst{s_0}{\vec{x}_0}{\vec{t}_0}\seq{\subst{s_1}{\vec{x}_1}{\vec{t}_1}, \dots, \subst{s_n}{\vec{x}_n}{\vec{t}_n}} \quad (\vec{x} = \bigoplus \vec{x}_i, \vec{t} = \bigoplus \vec{t}_i \text{ and } \vec{x}_i = \mathsf{occ}_{\vec{x}}(s_i) ) \]
\[   \subst{\la{\vec{y}} s}{\vec{x}}{\vec{t}} = \begin{cases} \la{\vec{y}} \subst{s}{\vec{x}}{\vec{t}}  	&  \text{ if }  \vec{x} \cap \vec{y} = \emptyset  ;\\ \la{\vec{y}} s &  \text{ otherwise.}  \end{cases} \]
\end{definition}

\subparagraph{Free Construction} 

\begin{definition}\label{app:embclosed}
Let $ \mathsf{E} $ be a symmetric closed multicategory and let $ i : \mathcal{L} \to \overline{\mathsf{E}} $ be a map of closed signatures. We define a family of maps $ \mathsf{RT}_{\gamma, \ty} : \rTerms(\mathcal{L}) (\gamma, \ty)  \to \mathsf{E}(i(\gamma), i(\ty)) $ by induction as follows: 

\scalebox{0.9}{\parbox{1.05\linewidth}{\[  \mathsf{RT}_{\ty, \ty}(x) = 1_{i(\ty)}  \qquad  \mathsf{RT}_{\gamma, \tyl \multimap \ty} (\la{\vec{x}} s) = \lambda ( \mathsf{RT}_{\gamma, \tyl, \ty} (s) ) \] \[ \mathsf{RT}_{ (\gamma_0, \dots, \gamma_k), \ty}(s \seqdots{t}{1}{k}) = ( ev \circ   \seq{\mathsf{RT}_{\gamma_0, \seqdots{\ty}{1}{k} \multimap \ty} (s), \mathsf{RT}_{\gamma_1, \ty_1}(t_1), \dots, \mathsf{RT}_{\gamma_k, \ty_1}(t_k)} )  \cdot \sigma .  \]}}
\end{definition}

\begin{proposition}\label{functsc}The following statements hold.\begin{enumerate}\item $ \mathsf{RT}_{\vec{\gamma}, \ty}(\subst{s}{\vec{x}}{\vec{t}}) = \mathsf{RT}_{\gamma, \ty}(s) \circ \seq{\mathsf{RT}_{\gamma_i, \ty_i}(t_i)}_{i \in \length{\gamma}} .$\item If $ s \to_{\mathsf{sc}} s' $ then  $  \mathsf{RT}_{\gamma, \ty}(s) = \mathsf{RT}_{\gamma, \ty}(s') .$\item If $s =_{\mathsf{sc}} s' $ then $  \mathsf{RT}_{\gamma, \ty}(s) = \mathsf{RT}_{\gamma, \ty}(s') .$
\end{enumerate}\end{proposition}
\begin{proof}
\begin{enumerate}
\item By induction on $ s .$ If $s  = x $ the proof is the same as in Proposition \ref{functrep}. If $ s = s_0 \seqdots{s}{1}{l} $ with $ (\delta_0^1, \dots, \delta^2_l) \cdot \sigma \vdash s : \ty, \delta_0^1, \vec{x}_0 : \tyl_0, \delta_0^2 \vdash s_0 : \seqdots{b}{1}{l} \multimap \ty , \delta_i^1, \vec{x}_i : \vec{a}_i, \delta_i^2 \vdash s_i : b_i  $ with  $\left(\bigoplus \vec{x}_i \right) \cdot \sigma = x_{1}, \dots, x_k, \left(\bigoplus \tyl_i \right) \cdot \sigma = \ty_{1}, \dots, \ty_k . $ By definition $ \subst{s}{\vec{x}}{\vec{t}} = \seq{\subst{s}{\vec{x}_1}{\vec{t}_1}, \dots, \subst{s_l}{\vec{x}_l}{\vec{t}_l}} .$ We can then apply the IH and conclude by associativity of composition. The abstraction case is a corollary of the IH and naturality of $\lambda(-)$.
\item By induction on the step. If it is $\beta , $ then it is a corollary of the former point of this lemma. If it is $ \eta, s \to_{\mathsf{sc}} \la{\vec{x} } s \vec{x}, $ we have that $ \mathsf{RT}_{\gamma, \tyl \multimap  \ty}(\la{\vec{x} } s \vec{x}) = \lambda (ev \circ \seq{s, id}) .$ We can conclude since $ \lambda $ is the inverse of $ ev \circ \seq{-, id} .$ The contextual cases are a direct application of the IH.
\item If $ s _{\mathsf{sc}} s' $, we observe that there exists $ n \in \mathbb{N} $ s.t. $ s = s_0 \leftrightarrow  s_1 \leftrightarrow \dots \leftrightarrow s_n = s' $ where $ \leftrightarrow $ stands for the reflexive and symmetric closure of $ \to_{\mathsf{sc}} .$ We prove the result by induction on $ n .$ If $ n = 0 $ the result is a direct corollary of the former point of the lemma. If $ n = p + 1 $ the result is a direct corollary of the IH.
\end{enumerate}
\end{proof}

\begin{theorem}[Free Construction]
Let $ \mathsf{S} $ be a a symmetric closed multicategory and $ i :  \mathcal{L} \to \overline{\mathsf{S}}  $ a map of representable signatures. There exists a unique symmetric closed functor $ i^{\ast} : \freescm{\mathcal{L}}  \to   \mathsf{S} $ such that $ \overline{i^\ast} \circ \eta_{\mathcal{L}} = i . $
\end{theorem}
\begin{proof}
The functor is defined exploiting Definition \ref{app:embclosed}.  By the former proposition, $ i^\star$ is well-defined and preserves composition and the identities. Given another functor $ i' : \freerm{\mathcal{R}} \to \mathsf{S} $ s.t.   $i = \overline{i'} \circ \eta_{\mathcal{R}} ,$ one proves that $ i^\ast = i' $ pointwise, exploiting the fact that $i' $ has to preserve the representable structure.\end{proof}

\section{Autonomous Case}

 \begin{proposition}\label{emb}
We have an embedding (injective map preserving the inductive structure and substitution) $ \iota_{\spadesuit} : \mathsf{\Lambda}_{\spadesuit} (\mathcal{N})(\gamma; \ty )\hookrightarrow \autTerms (\mathsf{aut}(\mathcal{N}))(\gamma; \ty)  $ for $\spadesuit \in \{\mathsf{re}, \mathsf{reps}, \mathsf{sc} \} $.
\end{proposition}

We treat these embeddings as they were inclusions. The embeddings concerns also type derivations, meaning that we will identify representable or closed type derivations with appropriate autonomous ones. Given a term $\gamma \vdash t : \ty ,$ there is a unique type derivation for it.

\subparagraph{Type System}

\begin{definition}
Let $ s $ be an autonomous term. We define the set of its \emph{subterms} by induction on $ s$ as follows.
\[      \mathsf{ST}(x) = \{  x\} \qquad  \mathsf{ST}(\la{\vec{x}}s ) = \{ \la{\vec{y}}  s' \mid s' \in \mathsf{ST}(s)  \} \sqcup \{ \la{\vec{x}} s \} \] \[ \mathsf{ST}(\seqdots{s}{1}{k}) = \bigsqcup_{i \in [k]} \mathsf{ST}(s_i) \sqcup \{ \seqdots{s}{1}{k} \}                        \]
\[        \mathsf{ST}(pq) = \mathsf{ST}(p) \sqcup \mathsf{ST}(q) \sqcup \{ pq\} \qquad \mathsf{ST}(s[\vec{x}:=t]) = \mathsf{ST}(s) \sqcup \mathsf{ST}(t) \sqcup \{s[\vec{x}:=t] \}           \]
\end{definition}

\begin{proposition}[Canonicity of Typing]\label{aut:canonsym}
If $ \pi \triangleright \gamma \vdash s : \ty $ and $ \pi' \triangleright \gamma \vdash s : \ty' $ then $\ty = \ty' $ and $  \pi = \pi' .$
\end{proposition}
\begin{proof}
by induction on $s . $   In the cases where a merging of type contexts happens, such as the list case, we rely on the properties of shuffle permutations and on the fact that type contexts are \emph{linear}. Hence, the action of permutations on contexts is always fixedpoint-free. We prove the list case. Let $s = \seqdots{s}{1}{k}, $ with the following type derivations: 
\[          \pi = \begin{prooftree}\hypo{ \pi_i  }\ellipsis{}{\gamma_i \vdash s_i : \ty_i  }      \hypo{ \sigma \in \mathsf{shu}(\gamma_1, \dots, \gamma_k)  }\infer2{ \act{(\gamma_1, \dots, \gamma_k)}{\sigma} \vdash \seqdots{s}{1}{k} : \seqdots{\ty}{1}{k}  } \end{prooftree}                     \qquad   \pi' = \begin{prooftree}\hypo{ \pi'_i  }\ellipsis{}{\gamma'_i \vdash s_i : \ty'_i  }      \hypo{ \sigma' \in \mathsf{shu}(\gamma'_1, \dots, \gamma'_k)  }\infer2{ \act{(\gamma'_1, \dots, \gamma'_k)}{\sigma'} \vdash \seqdots{s}{1}{k} : \seqdots{\ty'}{1}{k}  } \end{prooftree}                    \]
with $\act{(\gamma_1, \dots, \gamma_k)}{\sigma} = \act{(\gamma'_1, \dots, \gamma'_k)}{\sigma'} .$ We prove that $ \gamma_i = \gamma'_i  $ for all $ i \in [k] . $ We know that $ \abs{\gamma_i} = \abs{\gamma'_i} ,$ since the variables appearing in both contexts are \emph{exactly} the free variables of $ s_i  ,$ by linearity. By contradiction, suppose that $ \gamma_i \neq \gamma'_i . $ This means that for some $ x, y \in \abs{\gamma_i} ,$ we have that $ \gamma_i(x) < \gamma_i (y) $ and $ \gamma'_i(y) <  \gamma'_i(x)  $ or $ \gamma_i(y) < \gamma_i (x) $ and $ \gamma'_i(x) <  \gamma'_i(y) .$ However, $\sigma  $ and $\sigma' $ are shuffles, meaning that if $ \gamma_i(x) < \gamma_i(y) $ or $ \gamma'_i (x) < \gamma'_i (y) $ then $ \sigma (\gamma_i (x)) < \sigma (\gamma_i (y))  $ and $\sigma (\gamma'_i (x)) < \sigma (\gamma'_i (y))  ,$ contradiction. Hence $\gamma_i = \gamma_j $ and $ \sigma = \sigma ' ,$ being the action of permutation fixedpoint-free, since the contexts have no repetitions of variables. We can the apply the IH, get $\pi_i = \pi'_i  $ and conclude.
\end{proof}

\subparagraph{Substitution and Reduction}
\begin{figure}[t]
\begin{align*}
& \ctx ::= \hole{\cdot} \mid  \seq{s_1, \dots, \ctx, \dots,  s_n}   \mid   \ctx [\vec{x} := t] \mid s [\vec{x} := \ctx] \mid \ctx s \mid s \ctx \mid \la{\vec{x}} \ctx \mid  f (s_1, \dots, \ctx, \dots, s_n)\\
& \ctxe ::= \hole{\cdot} \mid   \seq{s_1, \dots, \ctxe, \dots, s_n} \mid \ctxe [\vec{x} := s]  \mid s [\vec{x} := \ctxe] \quad (\ctxe \neq \hole{\cdot} ) \mid \ctxe s \mid s \ctxe \mid \la{\vec{x}} \ctxe   \mid f (s_1, \dots, \ctxe, \dots, s_n)\\
 & \mathsf{D} ::= \hole{\cdot} \mid   \seq{s_1, \dots, \mathsf{D}, \dots, s_n} \mid \mathsf{D} [\vec{x} := s] \mid s [\vec{x} := \mathsf{D}]  \mid s \mathsf{D} \mid \mathsf{D} t \quad (\mathsf{D} \neq \hole{\cdot}) \mid \la{\vec{x}} \mathsf{D} \mid f s_1, \dots, \mathsf{D}, \dots, s_n)\\
 &\ctxl ::= \hole{\cdot} \mid \ctxl [\vec{x} := t]
\end{align*}
\caption{Autonomous contexts with hole.}\hrulefill
\end{figure}

\begin{figure}[!t]
	\centering 		 	

			\begin{align*} 
			&\text{\emph{$\beta_1$ Root-Step}: }	   s [x^{\ty_1}_1, \dots, x^{\ty_k}_k := \ctxl[\seqdots{t}{1}{k}]] \to_{\beta 1}  \ctxl[\subst{s}{x_1, \dots, x_k}{t_1, \dots, t_k}  ]   
			\\
	&\text{\emph{$\beta_2$ Root-Step}: }	  \ctxl[\la{\seq{x_1^{\ty_1}, \dots, x_k^{\ty_k}}} s] \seqdots{t}{1}{k} \to_{\beta 2}  \ctxl[\subst{s}{x_1, \dots, x_k}{t_1, \dots, t_k}]            
		.	\\
	&		\text{\emph{$\eta_1$ Root-Step}: }	 s \to_{ \eta 1} \vec{x}[\vec{x}^{\tyl} := s]                              \qquad \text{where } \vec{x}  \text{ fresh }, \ \gamma \vdash s : \tyl , \ s \notin \mathsf{LT}. \\
	&		  \text{\emph{$\eta_2$ Root-Step}: }      s \to_{\eta 2}     \la{\vec{x}^{\tyl}} (s \vec{x})       \qquad \text{where } \vec{x}  \text{ fresh }, \gamma \vdash s : \tyl \multimap \ty, s \notin \mathsf{AT}.
			\\
	& \text{\emph{Contextual extensions}:}  \begin{prooftree} \hypo{  s \to_{\beta 1} s'  }\infer1{  \ctx[s] \to_{\beta 1} \ctx[s']   } \end{prooftree} \qquad  \begin{prooftree} \hypo{  s \to_{\beta 2} s'  }\infer1{  \ctx[s] \to_{\beta 2} \ctx[s']   } \end{prooftree} \qquad  \begin{prooftree} \hypo{  s \to_{\eta 1} s'  }\infer1{  \ctxe[s] \to_{\eta 1} \ctxe[s']   } \end{prooftree} \qquad  \begin{prooftree} \hypo{  s \to_{\eta 2} s'  }\infer1{  \mathsf{D}[s] \to_{\eta 2} \mathsf{D}[s']   } \end{prooftree} \\ & \qquad \qquad {\to_\beta} = {\to_{\beta 1} \cup \to_{\beta 2} } \qquad {\to_\eta} = {\to_{\eta 1} \cup \to_{\eta 2} } \qquad  {\toaut} = {\to_{\beta}   \cup \to_{\eta}}  .& \\
		 & \text{   \emph{Structural equivalence}:} \quad \ctx[s [\vec{x} := t]  ] \streq \ctx[s] [\vec{x} := t] \qquad \vec{x} \notin \fv{\ctx}    .
\end{align*}	\caption{Autonomous reduction relations and structural equivalence.} 
\hrulefill
	\label{fig:toaut}
\end{figure}

\begin{definition}
Let $ s \in \autTerms $ and $   x_1,\ dots, x_k \in \mathcal{V} . $ We define the list of \emph{occurrences} of $x_1, \dots, x_ks $ in $s $ by induction as follows: 
\[   \mathsf{occ}_{\vec{x}} (x) = \begin{cases}    x & \text{ if } x\in \vec{x} ; \\ \seq{} & \text{ otherwise.}       \end{cases}      \qquad  \mathsf{occ}_{\vec{x}}(\seqdots{s}{1}{k}) = \bigoplus \mathsf{occ}_{\vec{x}} (s_i )    \] \[\mathsf{occ}_{\vec{x}}(s[\vec{y} := u]) = \mathsf{occ}_{\vec{x}} (s) :: \mathsf{occ}_{\vec{x}}(u)   \qquad    \mathsf{occ}_{\vec{x}}(s t) = \mathsf{occ}_{\vec{x}}(s) \oplus  \mathsf{occ}_{\vec{x}} ( t ) \qquad \mathsf{occ}_{\vec{x}}( \la{\vec{y} } s) = \mathsf{occ}_{\vec{x}} (s) .  \]
\end{definition}

From now on, whenever we deal with linear substitutions we shall always assume that they are well-defined, hence the substituted variables are contained in the free variables of the considered term.

\begin{definition}[Linear Substitution]
Let $  s,  t_1, \dots, t_k \in \autTerms $ and $ \  x_1, \dots, x_k  \subseteq \fv{s}. $ We define the \emph{($k $-ary linear) substitution} of $ x_1, \dots, x_k$ by $ t_1, \dots, t_k $ in $s $ by induction as follows:
\[     \subst{x}{x}{t} = t \qquad \subst{x}{}{} = x \quad (\text{ if } k = 0) \] \[     \subst{\seqdots{s}{1}{n}}{\vec{x}}{\vec{t}} =  \seq{\subst{s_1}{\vec{x}_1}{\vec{t}_1}, \dots, \subst{s_n}{\vec{x}_n}{\vec{t}_n}} \quad (\vec{x} = \bigoplus \vec{x}_i, \vec{t} = \bigoplus \vec{t}_i \text{ and } \vec{x}_i = \mathsf{occ}_{\vec{x}}(s_i) ) \]
\[   \subst{s_1[\vec{y} := s_2 ]}{\vec{x}}{\vec{t}} = \subst{s_1}{\vec{x}_1}{\vec{u}_1}[\vec{y} := \subst{s_2}{\vec{x}_2}{\vec{t}_2}]  \quad (\vec{x} = \vec{x}_1 :: \vec{x}_2, \vec{t} = \vec{t}_1 :: \vec{t}_2 \text{ and } \vec{x}_i = \mathsf{occ}_{\vec{x}}(s_i)  \]\[     \subst{s_1 s_2}{\vec{x}}{\vec{t}} =  \subst{s_1}{\vec{x}_0}{\vec{t}_1}\subst{s_2}{\vec{x}_2}{\vec{t}_2} \quad (\vec{x} = \vec{x}_1 \oplus \vec{x}_2, \vec{t} = \vec{t}_1 \oplus \vec{t}_2 \text{ and } \vec{x}_i = \mathsf{occ}_{\vec{x}}(s_i) ) \]
\[   \subst{\la{\vec{y}} s}{\vec{x}}{\vec{t}} = \begin{cases} \la{\vec{y}} \subst{s}{\vec{x}}{\vec{t}}  	&  \text{ if }  \vec{x} \cap \vec{y} = \emptyset  ;\\ \la{\vec{y}} s &  \text{ otherwise.}  \end{cases} \]
\end{definition}

 Contexts and reduction relation $\toaut$, together with its subreductions $\beta$ and $ \eta  $  are defined by putting together the reductions $ \torp $ (Figure \ref{fig:torp}) and $\to_{\mathsf{sc}} $ (Figure \ref{fig:toc}). The embedding of representable terms and symmetric closed terms into autonomous ones preserves the reduction.

\begin{proposition}\label{embred} Let $ \clubsuit \in \{   \beta, \eta \} ,$ $s, s' \in \repTerms (\mathcal{R})  $ and $t, t' \in \rTerms (\mathcal{L})  .$ $ s \torp^{\clubsuit} s' $  iff $ s \toaut^{\clubsuit} s' $   and $ t \to_{\mathsf{sc}}^{\clubsuit} t' $  iff $ t \toaut^{\clubsuit} t' .$
\end{proposition}

The autonomous reductions are defined in Figure \ref{fig:toaut}. In the structural equivalence rule, we assume that $ \ctx$ does not bind variables of $ t.$ The structural equivalence is defined as the smallest congruence generated by the rule in Figure \ref{fig:toaut}.

\begin{lemma}[Subject Substitution]\label{aut:subjsubst}
Let $ \gamma_i \vdash t_i : \ty_i $ and $ \delta_1, x_1 : \ty_1, \dots, \ty_k, \delta_2 \vdash s : \ty . $  We have that $ \delta_1,  \gamma_1, \dots, \gamma_k, \delta_2  \vdash \subst{s}{x_1, \dots, x_k}{t_1, \dots, t_k} : \ty .$
\end{lemma}
\begin{proof}
By induction on the structure of $ s.$ The variable case is immediate.
 
let $ s = \la{\vec{y}} s' . $ By definition we have $ \subst{\la{\vec{y}} s'}{\vec{x}}{\vec{t}} = \la{\vec{y}} \subst{s'}{\vec{x}}{\vec{t}}. $ We then apply the IH and conclude.

Let $ s = \seqdots{s}{1}{k} $ with $ \zeta_i \vdash s_i : b_i $ for $i \in [n] $ and $  \zeta_1, \dots, \zeta_n =  \delta_1, x_1 : \ty_1, \dots, \ty_k, \delta_2 .$ Hence, there exist decompositions  $ x_1, \dots, x_k = \vec{x}_1 :: \dots :: \vec{x}_n, \ty_1, \dots, \ty_k = \tyl_1 :: \dots :: \tyl_n $ and $ \vec{x}_j \in \fv{s_j}  $ for all $ j \in [n] $ s.t. $  \zeta_j = (\zeta^{1}_j, \vec{x}_j : \tyl_j , \zeta^2_j ) \vdash s_j : b_j .  $  By IH we have that $ \zeta^1_j, \vec{\gamma}_j, \zeta^2_j \vdash \subst{s_j}{\vec{x}_j}{\vec{t}_j} : b_j .  $ By definition of substitution we have $\subst{s}{\vec{x}}{\vec{t}} = \seq{\subst{s_1}{\vec{x}_1}{\vec{t}_1}, \dots, \subst{s_n}{\vec{x}_n}{\vec{t}_n}} $. Then $  \zeta^1_1, \vec{\gamma}_1, \zeta^2_1, \dots, \zeta^1_n, \vec{\gamma}_n, \zeta^2_n \vdash \subst{s}{\vec{x}}{\vec{t}} : \ty .$ By definition of typing, we know that $ (\delta_1,  \gamma_1, \dots, \gamma_k, \delta_2 ) = (  \zeta^1_1, \vec{\gamma}_1, \zeta^2_1, \dots, \zeta^1_n, \vec{\gamma}_n, \zeta^2_n ) .  $ We can then conclude.

The other cases follow a pattern similar to the list case.
\end{proof}

\begin{theorem}[Subject Reduction and Equivalence]\label{aut:subjred}
Let $ s \toaut s' $ or $ s \streq s' $ and $ s \in \autTerms(\mathcal{A})(\gamma; \ty) .$ Then $ \gamma \vdash_{\mathsf{aut}} s' : \ty .$
\end{theorem}
\begin{proof}
By induction on $ s \toaut s' $. The base case is a corollary of the former lemma, the other cases are a direct application of the IH.
\end{proof}

\begin{lemma}
We have that $ \subst{\subst{s}{\vec{x}}{\vec{t}}}{\vec{y}}{\vec{u}} = \subst{\subst{s}{\vec{y}_1}{\vec{u}_1}}{\vec{x}}{\subst{\vec{t}}{\vec{y}_2}{\vec{u}_2}}   $ with $ \vec{u} = \vec{u}_1 :: \vec{u}_2 , \vec{y} = \vec{y}_1 :: \vec{y_2}$.
\end{lemma}
\begin{proof}
By induction on $s . $ If $s = x  $ and $ \vec{x} = x , \vec{t} = t $ then $\vec{y} \in \fv{t} $ and by definition we can conclude. If $ s = \seqdots{s}{1}{k} ,$ by definition $ \subst{s}{\vec{x}}{\vec{t}} = \seq{\subst{s}{\vec{x}_1}{\vec{t}_1}, \dots, \subst{s_k}{\vec{x}_k}{\vec{t}_k}} . $ Again, by definition of substitution $\subst{\subst{s}{\vec{x}}{\vec{t}}}{\vec{y}}{\vec{u}} =  \seq{\subst{\subst{s_1}{\vec{x}_1}{\vec{t}_1}}{\vec{y}_1}{\vec{u}_1} , \dots, \subst{\subst{s_k}{\vec{x}_k}{\vec{t}_k}}{\vec{y}_k}{\vec{u}_k} }.$ We can then apply the IH and conclude. The explicit substitution case follows a similar pattern. 
\end{proof}

\begin{lemma}\label{aut:repeqsub} The following statements hold. If $ s \toaut s' $ then $ \subst{s}{\vec{x}}{\vec{t}} \toaut^\ast \subst{s'}{\vec{x}}{\vec{t}} $, if $ \vec{t} \toaut \vec{t}' $ then $ \subst{s}{\vec{x}}{\vec{t}} \toaut^\ast \subst{s}{\vec{x}}{\vec{t}'} ; $ if $ s =_{\mathsf{aut}} s'  $ and $ \vec{t} =_{\mathsf{aut}} \vec{t}' $ then $ \subst{s}{\vec{x}}{\vec{t}}  =_{\mathsf{aut}} \subst{s'}{\vec{x}}{\vec{t'}} ;$ if $ s \streq s'  $ and $ \vec{t} \streq \vec{t}' $ then $ \subst{s}{\vec{x}}{\vec{t}} \streq \subst{s'}{\vec{x}}{\vec{t'}} .$ 
\end{lemma} 
\begin{proof}
\begin{enumerate}
\item By induction on $ s \toaut s' . $ Let $ s = p[\vec{y} := \vec{u}] $ and $ s' = \subst{p}{\vec{y}}{\vec{u}} .$ We need to prove that $ \subst{p[\vec{y} := \vec{u}]}{\vec{x}}{\vec{t}} \toaut  \subst{\subst{p}{\vec{y}}{\vec{u}}}{\vec{x}}{\vec{t}}  . $ By definition of linear substitution, we have $ \subst{p[\vec{y} := \vec{u}]}{\vec{x}}{\vec{t}}  = \subst{p}{\vec{x}_1}{\vec{t}_1}[\vec{y} := \subst{\vec{q}}{\vec{x}_2}{\vec{t}_2}]$ and $\subst{p}{\vec{x}_1}{\vec{t}_1}[\vec{y} = \subst{\vec{q}}{\vec{x}_2}{\vec{t}_2}] \toaut \subst{\subst{p}{\vec{x}_1}{\vec{t}_1}}{\vec{y}}{\subst{\vec{q}}{\vec{x}_2}{\vec{t}_2}} . $ By the former lemma, we can conclude that \[ \subst{\subst{p}{\vec{x}_1}{\vec{t}_1}}{\vec{y}}{\subst{\vec{q}}{\vec{x}_2}{\vec{t}_2}} =  \subst{\subst{p}{\vec{y}}{\vec{u}}}{\vec{x}}{\vec{t}}.\]  The other cases are a direct application of the IH. 
\item By induction on $ s .$ If $s  =x  $ and $ \vec{t} = t  $ then $\subst{s}{\vec{x}}{\vec{t}} = t $ and the result is immediate. The other cases are direct applications of the IH.
\item We remark that if $ s \repaut s' $ then there exists $ n \in \mathbb{N} $ and terms $ s_i $ s.t. $    s = s_0  \leftrightarrow_{\mathsf{re}}  s_1 \dots  \leftrightarrow_{\mathsf{re}} s_n = s'        , $ where $\leftrightarrow $ stands for the reflexive symmetric closure of $\torp . $ We prove the result by induction on $ n. $ If $ n = 0 ,$ we need to prove that if $ \vec{t} \repeq \vec{t}' $ then $ \subst{s}{\vec{x}}{\vec{t}} \repeq \subst{s}{\vec{x}}{\vec{t}'} . $ We do it by induction on $m $ s.t. $  \vec{t} = \vec{t}_0  \leftrightarrow_{\mathsf{re}}  \vec{t}_1 \dots  \leftrightarrow_{\mathsf{re}} \vec{t}_m = \vec{t}' .        $ If $ m = 0$ the result is immediate by reflexivity. If $ m = p + 1 , $ by IH $ \subst{s}{\vec{x}}{\vec{t}} \repaut \subst{s}{\vec{x}}{\vec{t}_p} .$ Since $  \vec{t}_p \leftrightarrow \vec{t}_{p + 1} , $ by definition this means that either $ \vec{t}_p = \vec{t}_{p + 1} $ or $ \vec{t}_p \toaut \vec{t}_{p + 1} .$ The first case is immediate, the second is a direct corollary of the former point of this lemma. If $ n = p + 1 $, the result is a direct corollary of the first point of this lemma.
\end{enumerate}
\end{proof}

 We now prove that the structural equivalence is a strong bisimulation for the reduction $ \torp . $ Intuitively, this means that the equivalence does not affect the rewriting of terms.
 
\begin{lemma}
We have that $    \subst{p [\vec{y} := q]}{\vec{x}}{\vec{t}} \streq       \subst{p }{\vec{x}}{\vec{t}} [\vec{y} := q]$ and $\subst{p }{\vec{x}}{\vec{t}[\vec{y} := q]}  \streq       \subst{p }{\vec{x}}{\vec{t}} [\vec{y} := q]$
\end{lemma}
 
 \begin{proposition}
 If $ s' \streq s  $ and $ s \toaut t $ there exists a term $ t'$ s.t.  $ t \streq t' $ and  $ s' \torp t' . $
 \end{proposition}
 \begin{proof}
 By induction on $ s \streq t . $ 
 
 Let $ s' = \la{\vec{x}} (p[\vec{y} := q] ) $ and $ s =  (\la{\vec{x}} p)[\vec{y} := q] $ and $  (\la{\vec{x}} p)[\vec{y} := q] \toaut t.$ We reason by cases on the reduction step. If $q = \vec{q}  $ and $ t = \subst{(\la{\vec{x}} p)}{\vec{y}}{\vec{q}} $ then we set $ t' =  \la{\vec{x}} (\subst{ p}{\vec{y}}{\vec{q}} ). $ The other possible cases are necessarily contextual and directly follow from the IH, since we cannot apply an $ \eta$ ground step on abstractions. 
 
 Let $ s' = \seq{s_1, \dots, s_i [\vec{y}:= q], \dots, s_k} $ and $ s = \seqdots{s}{1}{k} [\vec{y} :=q] . $ Then either $ s_j \toaut p $ and $ t = \seq{s_1, \dots, p, \dots, s_k}[\vec{y} :=q] $ or $ q \toaut q'  $ and $ t = \seqdots{s}{1}{k} [\vec{y} :=q'],$ since we cannot apply $ \beta$ or $ \eta$ ground steps on list terms. In both cases the result is then a direct consequence of the IH.
 
 Let $ s' = p [\vec{y} := q]  v $ and $ s = (pv) [\vec{y} := q] . $ If $ p =\la{\vec{x}} p' , b = \vec{v}$ and $ t = \subst{p'}{\vec{x}}{\vec{v}} [\vec{y} := q] $ then we set $  t' = t . $ If $ \gamma \vdash s : \tyl \multimap \ty $ and $  t =  ((pv) [\vec{y} := q]) \vec{x} $ then by Theorem \ref{aut:subjred} $ \gamma \vdash s' : \tyl \multimap \ty $ and we set
  $ t' = ( p [\vec{y} := q]  v )\vec{x} .$ If $\gamma \vdash s : \tyl $ and $ t =  x [ x :=  ((pv) [\vec{y} := q])]  $ then by Theorem \ref{aut:subjred} we have $ \gamma \vdash s' : \tyl$ and we set $ t' =  x [ x :=  (p[\vec{y} := q]v)] .$ The other cases are contextual and follows directly from the IH.
  
  The other application case follows a similar pattern, the only thing to check is that $ \subst{s}{\vec{x}}{\vec{t} [\vec{y} := \vec{q}]} \streq  \subst{s}{\vec{x}}{\vec{t} } [\vec{y} := \vec{q}] $, which is a corollary of the former lemma.
  
 The explicit substitution cases follow a similar patter to the application one, applying the former lemma for the case where $  s \toaut t $ is a $\beta $ ground step.
 
 The contextual cases are direct consequences of the IH.
 \end{proof}

\begin{lemma}\label{aut:sizesub} Let $ s, t_1, \dots, t_k \in \repTerms $ and $x_1,\dots, x_k \in\fv{s} . $ We have that \[  \size{ \subst{s}{x_1, \dots, x_k}{t_1, \dots, t_k}  } = (\size{s} + \sum \size{t_i}) - k .         \] \end{lemma}
\begin{proof}
By induction on $s . $ If $s = x $ then $ k =1 $ and $ \subst{s}{\vec{x}}{\vec{t}} = t $. Then $  \size{\subst{s}{\vec{x}}{\vec{t}}} = \size{t} =  \size{t} + \size{x} - 1 , $ since $\size{x} = 1. $ If $ s = \seqdots{s}{1}{n} ,$ we have $ \subst{s}{\vec{x}}{\vec{t}} = \seq{\subst{s_1}{\mathsf{occ}_{\vec{x}} (s_1)}{\vec{t}_1}, \dots, \subst{s_n}{\mathsf{occ}_{\vec{x}} (s_n)}{\vec{t}_n}} . $ By IH, $\size{\subst{s_i}{\mathsf{occ}_{\vec{x}} (s_i)}{\vec{t}_i} } = \size{s_i} + \size{\vec{t}_i} - \length{\mathsf{occ}_{\vec{x}}(s_i)} .  $ Since $ k = \sum \length{\mathsf{occ}_{\vec{x}}(s_i)}  $ we can conclude. The other cases follows a similar pattern.
\end{proof}

\begin{lemma}\label{aut:commutasub}
Let $ s, \vec{t}, q  $ be resource terms with $ \vec{x} \subseteq \fv{s}$ and $ \length{\vec{t}} = \length{\vec{x}}$. If $ \subst{s}{\vec{x}}{\vec{t}} \to_{\eta} p $ then there exist $ s', \vec{t}' $ s.t. $ s \to_{\eta} s', \vec{t} \to_{\eta} \vec{t}' ,  p = \subst{s'}{\vec{x}}{\vec{t}'}  .  $
\end{lemma}
\begin{proof}
By induction on the structure of $s .  $ If $s = x $ then $\vec{t} = \seq{t} $ and $\subst{s}{\vec{x}}{\vec{t}} = t . $ Then we set $ s' = x $ and $ \vec{t}' = \seq{p} . $ 

If $ s = \la{\vec{y}} v $ then $\subst{s}{\vec{x}}{\vec{t}} = \la{\vec{y}} (\subst{v}{\vec{x}}{\vec{t}}) . $ The result is then a direct application of the IH. 

If $ s = v q $ then $\subst{s}{\vec{x}}{\vec{t}} =  \subst{v}{\mathsf{occ}_{\vec{x}(v)}}{\vec{t}_1} \subst{q}{\mathsf{occ}_{\vec{x}(q)}}{\vec{t}_2} . $ We have three possible cases. 
\begin{itemize}
\item If $  \subst{s}{\vec{x}}{\vec{t}} \to_{\eta} \la{\vec{y}} \subst{s}{\vec{x}}{\vec{t}} \vec{y}$ then we set $ s ' = \la{\vec{y}} (vq) \vec{y}$ and $\vec{t}' = \vec{t} . $
\item If $\subst{s}{\vec{x}}{\vec{t}} \to_{\eta} \vec{x} [\vec{x} := \subst{s}{\vec{x}}{\vec{t}}] $ then we set $ s' = \vec{x}[\vec{x} := vq] $ and $ \vec{t}' = \vec{t}. $
\item If the $\eta $ step is a contextual case the result is a direct consequence of the IH.
\end{itemize}
The cases of substitution and list follow a similar pattern to the one of the application case.
\end{proof}

\begin{proposition} If $ s \to_{\beta} t \to_{\eta} t' $ there exists $ s'  $ s.t. $ s \to_{\eta} s' $ and $ s' \to_{\beta}^\ast t' . $ 
\end{proposition}
\begin{proof}
By induction on the step $ s \to_\beta t . $ If $  s = p[ \vec{x} := \ctxl[ \vec{q} ] ] \to_{\beta 1} t = \ctxl[\subst{p}{\vec{x}}{\vec{q}}] $ the result is a corollary of Lemma \ref{aut:commutasub}. The same happens for the case $ s \to_{\beta_2} t . $ We do the contextual case of application. If $ s=  pq \to_\beta t = p'q $ with $ p \to_\beta p' $ and $ p'q \to_\eta t' $ we have three possible cases.
\begin{itemize}
\item If $ p'q \to_\eta  \la{\vec{y}} (p'q) \vec{y} .  $ In that case we set $ s'  = \la{\vec{y}} (pq) \vec{y}.$
\item If $p' q \to_\eta \vec{x}[\vec{x} := p'q] . $ In that case we set $ s' = \vec{x}[\vec{x}:= pq] .$
\item If the $\eta $ step is a contextual case the result is a direct consequence of the IH.
\end{itemize}
\end{proof}

\begin{proposition}\label{aut:commuta}
If $ s \to_{\beta}^\ast t \to_{\eta}^\ast t' $ there exists $ s'  $ s.t. $ s \to_{\eta}^\ast s' $ and $ s' \to_{\beta}^\ast t' . $ 
\end{proposition}
\begin{proof}
The proof is by induction on the length of the reduction $s \to_{\beta}^\ast , $ by applying the former proposition.
\end{proof}

Given $\gamma \vdash s  : \ty $ we define two sets of autonomous type judgments as follows$\mathsf{EST}_1 (s) = \{\delta \vdash p : 	\ty \mid p \in  \mathsf{ST}(s) \setminus \mathsf{LT}  \text{ s.t. } \ctxe[\delta \vdash p : \ty] = s  \text{ for some context } \ctxe  \} $ 
 and  $\mathsf{EST}_" (s) = \{\delta \vdash p : 	\ty \mid p \in  \mathsf{ST}(s) \setminus \mathsf{LT}  \text{ s.t. } \mathsf{D}[\delta \vdash p : \ty] = s  \text{ for some context } \mathsf{D}  \} .$ We set $ \eta_1 (s) =  \sum_{\delta \vdash p : \ty \in \mathsf{EST}(s)} \size{\ty}_1 $ and $ \eta_2 (s) =  \sum_{\delta \vdash p : \ty \in \mathsf{EST}_2(s)} \size{\ty}_2  .  $

 \begin{remark}
 The intuition behind the sets $\mathsf{EST}_1 (s) $ and $\mathsf{EST}_2(s) $ is that we are considering all types of subterms of $s $ on which we could perform the $\eta $-reduction. The restrictions on the shape of $ p$ is indeed directly derived from the ones on the $\eta $-reductions. 
 \end{remark}

 \begin{proposition} The following statements hold.
If $ s \to_\beta s' $ then $ \size{s'} < \size{s} $;  if $  s \to_{\eta 1 } s'$ then $ \eta _1(t) < \eta (s)_1 $;  if   $s \to_{\eta 2 } s'$ then $ \eta _2(t) < \eta (s)_2.$  \end{proposition}
\begin{proof}
\begin{enumerate}
\item By induction on $  s \to_\beta s' ,$ exploiting Lemma \ref{aut:sizesub}. 
\item By induction on $s \to_{\eta 1} s' . $   Let $ \gamma \vdash s : \seqdots{\ty}{1}{k}  $ and $ s \to_{\eta 1}  \vec{ x} [\vec{x}= s]  $ where $ s \notin \mathsf{LT}. $  We have $\mathsf{EST}_1(\vec{ x} [\vec{x}= s]) = \mathsf{EST}_1 (s)  \setminus \{ s \} .$ We can then conclude, since \[  \eta_1 (s) =  \eta_1 ( \vec{ x} [\vec{x}= s]) + \size{\seqdots{\ty}{1}{k}}_1\] and $ \size{\seqdots{\ty}{1}{k}} \neq 0 $. Let $ s = \la{\vec{x}} t$ and $ s' = \la{\vec{x}} t'$ with $ t \to_{\eta 1} t'. $ By IH we have that $ \eta_1 (t') < \eta_1( t) . $ By subject reduction and definition of $ \eta_1 $ we can conclude. Let $ s =  t q  $ and $ s' = t' q $ with $ t \to_{\eta 1} t' .$ Again the result is a corollary of the IH and subject reduction, by observing that $ \eta_1 (tq)= \eta_1 (t) + \eta_1 (q)$. The same works for the substitution and list cases.  
\item By induction on $ s \to_{\eta 2} s' .$ Let $ \gamma \vdash s : \tyl \multimap \ty  $ and $ s \to_{\eta 2} \la{\vec{x}} s \vec{x} .  $ where $ s \notin \mathsf{LT} .$ We have $\mathsf{EST}_2(\la{\vec{x}} s \vec{x} ) = \mathsf{EST}_2 (s)  \setminus \{ s \} \sqcup \{ \vec{x}\} .$ We have that  $ \eta_2(s) = \eta_2 (\la{\vec{x}} s \vec{x}) + \size{\tyl \multimap \ty}_2 - \size{\tyl}_2 . $ Since $ \size{\tyl}_2  < \size{\tyl \multimap \ty}_2$ we can conclude.
We can then conclude. 
\end{enumerate}
\end{proof}

\begin{proposition}[Local Confluence]The reduction $ \toaut  $ is locally confluent.\end{proposition}\begin{proof} Let $ s \toaut t_1 $ and $ s \toaut t_2 . $ We prove the result by induction on the step $ s \toaut t_1 $ and by reasoning by cases on $ s \toaut t_2 . $  If $ s = s' [\vec{x} := \ctxl[\vec{t}]] $ with $ t_1 = \ctxl[\subst{s'}{\vec{x}}{\vec{t}}] $  we reason by cases on $ s \toaut t_2 . $ if $ t_2 =  s'' [\vec{x} := \ctxl[\vec{t}]] $ then $ \ctxl[\subst{s'}{\vec{x}}{\vec{t}}] \toaut^\ast \ctxl[\subst{s''}{\vec{x}}{\vec{t}} ] $  by Lemma \ref{aut:repeqsub}. If $t_2 = s [\vec{x} := \ctxl[\vec{t}']]  $ then $ \ctxl[\subst{s}{\vec{x}}{\vec{t}} ] \toaut^\ast ctxl[\subst{s}{\vec{x}}{\vec{t}}]  $ again by Lemma \ref{aut:repeqsub}. If $ s \toaut t_2 $ is an $ \eta$-rule, the result is immediate by contextuality.

If $  s = (\la{\vec{x}} p) \vec{q} $ and $ t_1 = \subst{p}{\vec{x}}{\vec{q}} $ we reason by cases on $ s \toaut t_2 . $ if $ \la{\vec{x}} p) \vec{q} \toaut \la{\vec{x}} p') \vec{q'},$ then the result is a corollary of Lemma \ref{aut:repeqsub}. If $ s \toaut t_2 $ is an $ \eta$-rule, the result is immediate by contextuality.

Let $\gamma \vdash s : \tyl \multimap \ty   $  with $ s \toaut \la{\vec{x}} s \vec{x} . $ Then if $ s \toaut s' , $ by contextuality of the reduction $\la{\vec{x}} s \vec{x} \toaut \la{\vec{x}} s' \vec{x} .$

Let $ \gamma \vdash s : \tyl $  with $ s \toaut \vec{x}[\vec{x} := s] ,$ we can again conclude by contextuality.
The contextual cases are direct consequences of the IH.
\end{proof}

\end{document}

%% file: macro.tex
\newcommand*{\streq}{\equiv}

\newcommand*{\Ob}[1]{\mathrm{ob}(#1)}

\newcommand{\Mon}{\mathsf{Mon}}

\newcommand{\MON}{\mathsf{Mon}}
\newcommand{\REP}{\mathsf{rep}}
\newcommand{\mcati}{\mathsf{M}}
\newcommand{\mocati}{\mathit{M}}

\newcommand{\mcatic}{\mathsf{ClosedM}}
\newcommand{\mcatir}{\mathsf{RepM}}
\newcommand{\mcatirs}{\mathsf{RepsM}}
\newcommand{\mauto}{\mathsf{autoM}}

\newcommand*\rTerms{\mathsf{\Lambda_{\mathsf{sc}}}}

\newcommand*\repTerms{\mathsf{\Lambda_{rep}}}
\newcommand*\repsTerms{\mathsf{\Lambda_{reps}}}
\newcommand*\autTerms{\mathsf{\Lambda_{aut}}}

\newcommand*\stabi[1]{\mathsf{Stab}(#1)}

\newcommand*\stricti[1]{\mathsf{strict}(#1)}
\newcommand*\ract[2]{#1^{#2}}
\newcommand*\atm{\mathsf{At}}

\newcommand*\freesrm[1]{\mathsf{SRM}(#1)}
\newcommand*\freeaut[1]{\mathsf{AUT}(#1)}

\newcommand*\freerm[1]{\mathsf{RM}(#1)}
\newcommand*\freescm[1]{\mathsf{SCM}(#1)}

\newcommand*\subst[3]{ #1 \sub{#3}{#2}}

\newcommand*{\sub}[2]{\{#1/#2\}}

\newcommand*{\NF}[1]{\mathsf{nf}(#1)}

\newcommand*{\fv}[1]{\mathsf{fv}(#1)}

\newcommand*{\length}[1]{\mathsf{len}(#1)}

\newcommand*{\abs}[1]{\lvert #1 \rvert}

\newcommand*{\morph}[1]{\text{arr}(#1)}

\newcommand*{\source}[1]{\mathsf{source}(#1)}
\newcommand*{\act}[2]{{#1} \cdot {#2}}

\newcommand*{\size}[1]{\mathtt{size}\left(#1\right)}

\newcommand*{\targ}[1]{\mathsf{trg}(#1)}

\newcommand*{\tens}[3]{(#1_{#2} \otimes \dots \otimes #1_{#3})  }

\newcommand*{\seq}[1]{\langle #1 \rangle}

\newcommand*{\seqdots}[3]{\seq{#1_{#2}, \dots, #1_{#3} } }

\newcommand*{\la}[1]{\lambda #1.}

\newcommand*{\hole}[1]{[#1]}

\newcommand*{\ctx}{\mathtt{C}}
\newcommand*{\ctxl}{\mathtt{L}}
\newcommand*{\ctxe}{\mathtt{E}}

\newcommand*{\repeq}{ =_{\mathsf{rep}} }
\newcommand*{\repaut}{ =_{\mathsf{aut}} }

\newcommand*{\toaut}{\to_{\mathsf{aut}} }

\newcommand*{\torp}{\to_{\mathsf{rep}} }

\newcommand*{\ty}{a}

\newcommand*{\tyl}{\vec{a}}

\makeatletter
\newcommand*{\doublerightarrow}[2]{\mathrel{
  \settowidth{\@tempdima}{$\scriptstyle#1$}
  \settowidth{\@tempdimb}{$\scriptstyle#2$}
  \ifdim\@tempdimb>\@tempdima \@tempdima=\@tempdimb\fi
  \mathop{\vcenter{
    \offinterlineskip\ialign{\hbox to\dimexpr\@tempdima+1em{##}\cr
    \rightarrowfill\cr\noalign{\kern.5ex}
    \rightarrowfill\cr}}}\limits^{\!#1}_{\!#2}}}
\newcommand*{\triplerightarrow}[1]{\mathrel{
  \settowidth{\@tempdima}{$\scriptstyle#1$}
  \mathop{\vcenter{
    \offinterlineskip\ialign{\hbox to\dimexpr\@tempdima+1em{##}\cr
    \rightarrowfill\cr\noalign{\kern.5ex}
    \rightarrowfill\cr\noalign{\kern.5ex}
    \rightarrowfill\cr}}}\limits^{\!#1}}}
\newcommand{\colim@}[2]{\vtop{\m@th\ialign{##\cr
    \hfil$#1\operator@font lim$\hfil\cr
    \noalign{\nointerlineskip\kern1.5\ex@}#2\cr
    \noalign{\nointerlineskip\kern-\ex@}\cr}}}
\newcommand{\colim}{%
  \mathop{\mathpalette\colim@{\rightarrowfill@\textstyle}}\nmlimits@
}
\makeatother